\def\doublecolumn{1} 

\PassOptionsToPackage{dvipsnames}{xcolor}
\if 1\doublecolumn
  \documentclass[journal]{IEEEtran}
\else
  \documentclass[12pt, draftclsnofoot, onecolumn, letterpaper]{IEEEtran}
\fi

\usepackage{mypackage}
\usepackage{myrevision}

\begin{document}

\title{Secure Multi-Hop Relaying in Large-Scale Space-Air-Ground-Sea Integrated Networks}

\author{Hyeonsu Lyu,~\IEEEmembership{Student Member,~IEEE}, Hyeonho Noh,~\IEEEmembership{Member,~IEEE},
Hyun Jong Yang,~\IEEEmembership{Senior Member,~IEEE}, and Kaushik Chowdhury,~\IEEEmembership{Fellow,~IEEE}
\thanks{
    Hyeonsu Lyu is with the Department of Electrical Engineering, POSTECH, Korea (email: hslyu4@postech.ac.kr).
    Hyeonho Noh is with Department of Information \& Communication Engineering, Hanbat National University, Korea (email: hhnoh@hanbat.ac.kr). Hyun Jong Yang (co-corresponding author) is with the Department of Electrical and Computer Engineering, Seoul National University, Korea (email: hjyang@snu.ac.kr).
    Kaushik Chowdhury (co-corresponding author) is a Chandra Family Endowed Distinguished Professor in the Department of Electrical and Computer Engineering, University of Texas at Austin, USA (email: kaushik@utexas.edu).
    }
}

\makeatletter
\def\ps@IEEEtitlepagestyle{%
\def\@oddhead{%
    \parbox{\textwidth}{
    \small This paper has been submitted to IEEE for possible publication.\hfill {\scriptsize \thepage}\\
    Copyright may be transferred without notice, after which this version may no longer be accessible.}}
    \def\@evenhead{}%
    \def\@oddfoot{}%
    \def\@evenfoot{}%
}
\makeatother



\maketitle

\begin{abstract}
As a key enabler of borderless and ubiquitous connectivity, space–air–ground–sea integrated networks (SAGSINs) are expected to be a cornerstone of 6G wireless communications.
However, the multi-tiered and global-scale nature of SAGSINs amplifies the security vulnerabilities, particularly due to the hidden, passive eavesdroppers distributed across the network.
This paper introduces a joint optimization framework for multi-hop relaying in SAGSINs that maximizes the minimum user throughput while ensuring a minimum strictly positive secure connection (SPSC) probability.
We first derive a closed-form expression for the SPSC probability and incorporate this into a cross-layer optimization framework that jointly optimizes radio resources and relay routes.
Specifically, we propose an $\mathcal{O}(1)$ optimal frequency allocation and power splitting strategy—dividing power levels of data transmission and cooperative jamming.
We then introduce a Monte-Carlo relay routing algorithm that closely approaches the performance of the numerical upper-bound method.
We validate our framework on testbeds built with real-world datasets.
All source code and data for reproducing the numerical experiments will be open-sourced \hl{via GitHub}.
\end{abstract}

\IEEEpeerreviewmaketitle

\section{Introduction}
\label{sec:Introduction}
 
\IEEEPARstart{T}he era of hyper-connectivity begins. 
SpaceX's reusable launch vehicles have paved the way for extending commercial satellite communications to a broader user base \cite{Kodheli21-SurvTut},
and initiatives such as Google's Project Loon \cite{Bellemare2020-nature} and Aerostar's Thunderhead \cite{Elamassie24-TVT} have established the technical foundations for high-altitude platforms (HAPs) networks.
Moreover, the advances in sea surface \cite{Li25-arxiv} and underwater base stations \cite{Zhang24-COMMAG}, other forms of borderless networks, signal a big paradigm shift: a global-scale integrated network \cite{Meng25-SurvTut, Guo22-SurvTut}.
NTNs break down national boundaries, lay the groundwork for an interplanetary internet, and provide communication infrastructure across all regions of the Earth. 
Ubiquitous connectivity in line with this trend has emerged as a primary objective for 6G wireless \cite{IMT2030}, driving active discussions on standards for non-terrestrial networks \cite{3gpp.23.700-10, 3gpp.26.501, 3gpp.38.821}.
The integration of these vertical, heterogeneous network layers stands to transcend the limitations of terrestrial networks.

To realize the merits of the space-air-ground(-sea) integrated networks (SAG(S)INs), various technical challenges have been addressed in both industry and academia \cite{Wang23-TWC, Han22-TCOM, Zhang24-IOTJ, Li23-CL, Wang23-GLOBECOM, kakati25-OJCOMS, Sun22-TCOM, Nguyen23-TAES}.
Consequently, numerous techniques have been proposed to combine the complementary strengths of heterogeneous platforms, 
achieving ubiquitous global coverage and strengthening link reliability \cite{Meng25-SurvTut}.
However, the wide communication coverage of SAGSIN and the participation of numerous heterogeneous nodes can intrinsically intensify the possibility of eavesdropper's (Eve's) intrusion \cite{Guo22-SurvTut, Dai22-IOTM}.
Safeguarding the integrity and confidentiality of SAGSIN is critical in scenarios where the secrecy requirements are exceptionally high, such as maritime energy transportation and military operations.
Moreover, Eves can achieve considerable SNR in SAGSINs as LEO and HAPs emit signals via line-of-sight channels over vast areas, exponentially expanding the attack surface.
Therefore, physical-layer security must be employed in SAGSIN to protect information of legitimate users.

Ensuring secrecy for network nodes in SAGSINs presents several technical challenges.
Users in the SAGSIN must traverse multiple base station nodes to access the core internet \cite{Xu24-WCOM}.
This indicates that secrecy must be maintained across multiple relay hops, each of which constitutes a potential vulnerability exploitable by adversaries \cite{Dong10-TSP}.
Numerous studies have explored secure communications in SAGSINs where HAPs serve as intermediate nodes to bridge satellites and ground terminals in two-hop relay systems \cite{Wang23-TWC, Han22-TCOM, Zhang24-IOTJ, Li23-CL}.
However, ensuring secure communications across multi-hop relays in SAGSINs still needs further investigation.

Delivering high throughput to users while ensuring security presents another significant challenge.
In SAGSINs, optimizing radio resources to maximize throughput under security constraints is inherently intractable as acquiring accurate channel state information (CSI) of hidden Eves is practically infeasible \cite{Su22-Sensors, Abdalla23-TVT}.
However, many studies focus on optimizing the relay itself for given CSI, while the joint optimization of relay and radio resource management (RRM) in SAGSINs remains an open problem.

These challenges motivate the following research question:
\begin{tcolorbox}[colframe=black, colback=white, height= \if 1\doublecolumn 1.1 \else 1.5 \fi cm, boxrule=0.4mm]
\begin{center}
    \vspace{-0.135cm}
    \if 1\doublecolumn 
    \textit{\textbf{How can high-throughput, secure SAGSIN relay be established without Eve's channel information? }}
    \else
    \textit{\textbf{How can high-throughput, secure SAGSIN relay be established \\ without Eve's channel information? }}
    \fi
\end{center}
\end{tcolorbox}
\noindent This research question has rarely been addressed within the context of secure SAGSIN research, as shown in Table~\ref{tab:related_works}.
Unlike prior works that only optimize secure relaying in NTNs, 
we introduce a joint optimization framework for multi-hop relaying that maximizes the minimum user throughput while ensuring a prescribed strictly positive secure connection (SPSC) probability against unknown or passive Eves.

\begin{table}[htb]
    \centering
    \caption{Summary of the related works in Sec.~\ref{Sec:Related Works}}
    \label{tab:related_works}
    \begin{tabularx}{\if 1\doublecolumn 1 \else .62\fi \columnwidth}{ccccccc}
        \toprule
              Ref.            & Archit. &  Eve.   & \hspace{-.2cm}   RA  \hspace{-.1cm}    &     PC  \hspace{-.1cm}  &  Jamming \hspace{-.3cm}   &   Relay    \\
        \cmidrule(l{2pt}r{2pt}){1-1} \cmidrule(l{2pt}r{2pt}){2-7}
        \arrayrulecolor{lightgray}
        \cite{Wang23-TWC}       &  SAGIN  & \Known  &  \Xmark  &   \Xmark   &   \Xmark   &  Two-hop   \\
        \cmidrule(l{2pt}r{2pt}){1-1} \cmidrule(l{2pt}r{2pt}){2-7}
        \cite{Han22-TCOM, Zhang24-IOTJ, Li23-CL}
                                &  SAGIN  & \Known  &  \Xmark  & \mycheck & \mycheck &  Two-hop   \\
        \cmidrule(l{2pt}r{2pt}){1-1} \cmidrule(l{2pt}r{2pt}){2-7}
        \cite{Wang23-GLOBECOM}  &  SAGIN  & \Known  &  \Xmark  & \mycheck &   \Xmark   &     -      \\
        \cmidrule(l{2pt}r{2pt}){1-1} \cmidrule(l{2pt}r{2pt}){2-7}
        \cite{kakati25-OJCOMS}  &  SAGIN  & \Known  &  \Xmark  & \mycheck &   \Xmark   &  Two-hop   \\
        \cmidrule(l{2pt}r{2pt}){1-1} \cmidrule(l{2pt}r{2pt}){2-7}
        \cite{Eiza23-WCNC, Luo24-TNSM}
                                &  SAGIN  & Unknown &  \Xmark  &   \Xmark   &   \Xmark   & Multi-hop  \\
        \cmidrule(l{2pt}r{2pt}){1-1} \cmidrule(l{2pt}r{2pt}){2-7}          
        \cite{Sbeiti16-TWC}                                    
                                & UAV net.& Unknown &  \Xmark  &   \Xmark   &   \Xmark   & Multi-hop  \\
        \cmidrule(l{2pt}r{2pt}){1-1} \cmidrule(l{2pt}r{2pt}){2-7}          
        \cite{Sun22-TCOM}       & UAV net.& \Known  &  \Xmark  &   \Xmark   & \mycheck & Multi-hop \\
        \cmidrule(l{2pt}r{2pt}){1-1} \cmidrule(l{2pt}r{2pt}){2-7}          
        \cite{Nguyen23-TAES}    &  SGIN   & \Known  &  \Xmark  &   \Xmark   & \mycheck & Multi-hop  \\
        \cmidrule(l{2pt}r{2pt}){1-1} \cmidrule(l{2pt}r{2pt}){2-7}          
        \cite{Guo18-Access}     &  SGIN   & \Known  &  \Xmark  &   \Xmark   &   \Xmark   & Multi-hop  \\
        \cmidrule[0.75pt](l{2pt}r{2pt}){1-7}
        \textbf{Ours} & SAGSIN & Unknown & \mycheck &  \mycheck & \mycheck  &  Multi-hop \\
        \arrayrulecolor{black}
        \bottomrule
    \end{tabularx}
    \vspace{-.2cm}
    \\
    \begin{flushleft}
    * RA: Frequency resource allocation, PC: Power control, \\
    ~~S(A)GIN: Space(-air)-ground integrated network.
    \end{flushleft}
    \vspace{-.1cm}
\end{table}

The contributions of this work are summarized as follows:
\begin{itemize}
    \item \textit{\textbf{Cross-layer secure relaying framework}}:
    We propose a cross-layer framework for secure multi-hop relaying in large-scale SAGSINs that optimizes the max-min throughput \textbf{without Eve’s channel information}, 
    \hl{bridging two separate research fields---the secrecy outage probability and radio resource scheduling---with a new form of system model for secure communications that has been previously pioneered.}
    This formulation facilitates the optimization process compared to conventional max-min secrecy rate approaches, while achieving comparable performance to the max-min secrecy problem (Fig.~\ref{fig:mmsr_spsc}).
    \item \textit{\textbf{Analysis for SPSC constraint}}:
    \hl{
    We derive a closed-form approximation for the SPSC probability using stochastic geometry, 
    showing that it can be characterized as a function of Eve density, jamming power, and link distance.
    This generalizes the prior works \cite{Yao16-TCOM, Yao18-CL, Yao18-TWC}, which obtained closed-form expressions for the SPSC probability in the absence of jamming. 
    Furthermore, we rigorously analyze the approximation gap in Appendix~\ref{Supple:Derivation of the Bound Gap in the SPSC Approximation} and propose a calibration method to mitigate this gap in Figs.~\ref{fig:spsc_density}-\ref{fig:jam_vs_distance}, which has not been addressed in earlier studies.
    }
    \item \textit{\textbf{Efficient solution design}}: 
    We develop globally optimal closed-form solutions for radio resource management with $\mathcal{O}(1)$ computational complexity (\textbf{Theorem 1}), achieving the optimal frequency allocation, transmission, and jamming power control for a given routing topology.  
    We then propose a Monte-Carlo relay routing algorithm (Alg.~1) with $\mathcal{O}(KN\log N)$ computational complexity for $N$ nodes.  
    Our method consistently achieves a max-min throughput within approximately 5\% of the upper bound under various secure SAGSIN simulation settings (Figs.~\ref{fig:mmf_spsc}–\ref{fig:mmf_power}).
    \item \textit{\textbf{Real-world SAGSIN testbed}}: 
    We validate the proposed framework and solver using real-world data of terrestrial, HAPs, LEO satellites, and vessel data.
    To the best of our knowledge, this is \textbf{the first real-world SAGSIN testbed \hl{built upon a real-world dataset}} including HAPs base stations.  
    The demonstrations on the testbeds not only demonstrate the practicality and superiority of the RRM and routing algorithm (Fig.~\ref{fig:real_demo}), but also illustrate that the framework can immediately adapt to changes in the security environment (Fig.~\ref{fig:real_demo_us}).
\end{itemize}

\section{Related Works}
\label{Sec:Related Works}
Recent studies have explored diverse approaches to secure communication in SAGSIN.
We survey both general physical-layer security strategies and secure routing approaches in the broader context of NTNs, including LEO and HAP networks.
These related works are summarized in Table.~\ref{tab:related_works}.

\subsection{Secure Communications in SAGSINs} 
\cite{Wang23-TWC} proposes an aerial bridge scheme using UAVs to form secure tunnels and reduce the eavesdropping footprint.
\cite{Han22-TCOM} enhances secrecy rate via a joint UAV deployment and power allocation using successive convex approximation to decouple complex multi-layer constraints,
and \cite{Zhang24-IOTJ} further enhances security by optimizing UAV trajectory and power with a hybrid FSO/RF scheme that mitigates RF vulnerabilities through optical satellite links.
Similarly, \cite{Li23-CL} combines UAV-assisted non-orthogonal multiple access with cooperative jamming to counter both internal and external eavesdropping threats in SAGINs.
\cite{Wang23-GLOBECOM} proposes a label‐free deep learning framework that dynamically selects access strategies to maximize sum secrecy rate.
\cite{kakati25-OJCOMS} optimizes HAPs trajectories and beamforming by employing a generative AI-based DRL framework for proactive adaptation and improved secrecy energy efficiency under dynamic channel changes.

These studies have contributed to improving various secrecy utilities in small-scale SAGSINs with less than two-hops under the perfect knowledge of Eves.
However, it is essential to investigate wide-area multi-hop communications under imperfect knowledge of Eve’s channel, as SAGSINs are fundamentally intended to enable ubiquitous global connectivity.

\subsection{Secure Routing in NTNs}
\cite{Eiza23-WCNC} presents a hybrid SDN architecture with hierarchical control to enable secure and QoS-aware routing in SAGINs. 
\cite{Luo24-TNSM} design an artificial bee colony-based routing protocol for UAV networks, significantly improving routing performance and consensus efficiency in dynamic environments.
Meanwhile, \cite{Sbeiti16-TWC} proposes a position-aware routing protocol for airborne mesh networks to efficiently mitigate threats from wormhole and blackhole attacks.
Furthermore, \cite{Sun22-TCOM} propose a secure, energy-efficient UAV relay framework using collaborative beamforming to jointly enhance secrecy and reduce propulsion energy consumption.
Additionally, \cite{Nguyen23-TAES} analyzes security–reliability trade-offs in satellite–terrestrial relay networks with imperfect CSI by introducing a friendly jammer, deriving closed-form expressions for outage and intercept probabilities.
\cite{Guo18-Access} jointly designs secure relay selection and user scheduling for hybrid satellite–terrestrial relay networks, deriving closed-form expressions for average secrecy capacity.

These studies have enabled secure communication over broader geographic areas through multi-hop relaying. 
However, these studies may not achieve optimal network utility as they adopted physical channels for routing, or do not consider the trade-off between transmission and jamming power.
\hl{Moreover, pioneering studies on secrecy outage-aware routing with jammers have been conducted in \cite{Xu21-TWC, Xu21-IoTJ, Yao18-TWC}. However, these works have neither investigated the problem within the context of SAGSIN nor jointly considered physical-layer radio resources.}
Therefore, a cross-layer design to jointly optimize RRM and relay routing is required in multi-hop SAGSINs.

\section{System model}
\label{sec:System model}
\begin{figure*}[htb]
    \centering
    \includegraphics[width=\textwidth]{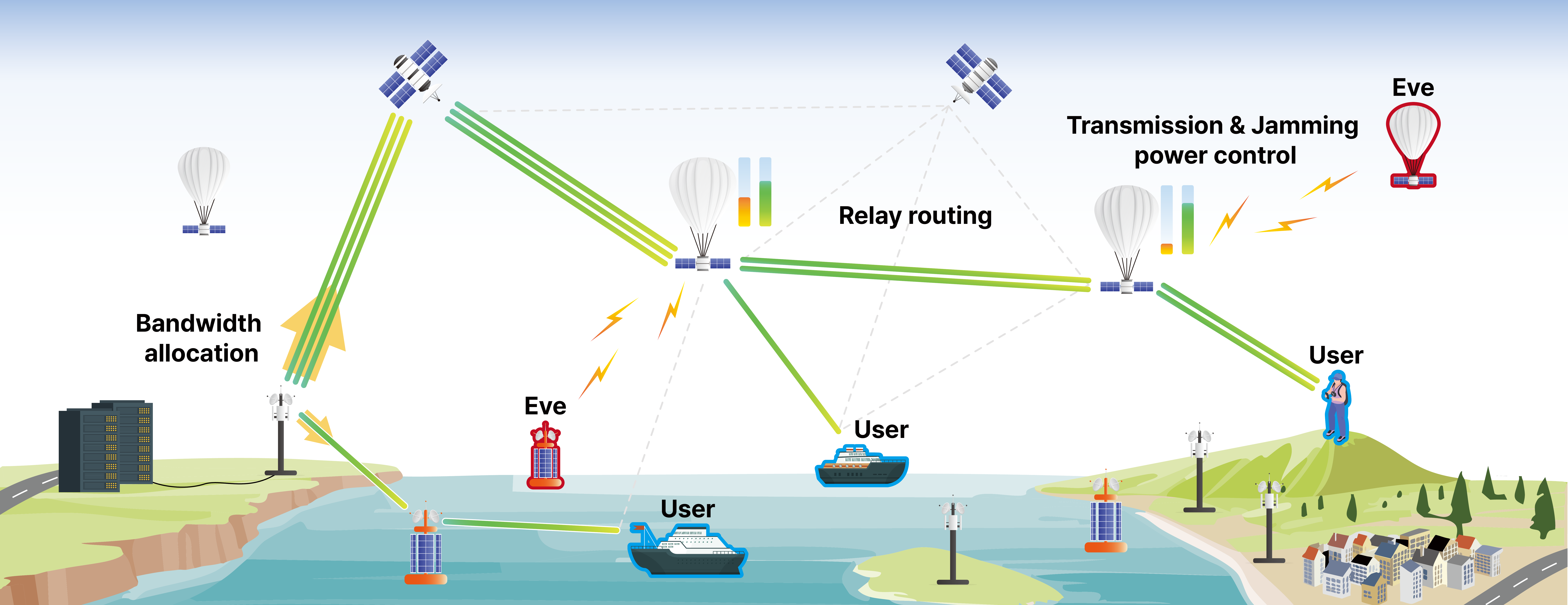}
    \caption{Illustration of the target scenario in secure multi‐hop relaying for SAGSINs. 
    The diagram shows a source node relaying data to multiple users via intermediate nodes. 
    At each relay hop, the nodes transmit both data and predefined jamming signals to degrade the Eves’ signal quality, whose locations are modeled as uniformly distributed across the space, air, ground, and sea layer.
    Each node dynamically allocates bandwidth between data transmission and jamming to maximize the minimum user throughput.
    }
    \label{fig:system_model}
\end{figure*}

\subsection{Scenario Description}

This paper aims to construct a secure relaying route in the SAGSIN, with consideration of radio resources, keeping the SPSC probability above a certain threshold in the presence of the hidden Eves.
Fig.~\ref{fig:system_model} illustrates the target scenario with three control variables.
The SAGSIN relay adopts the decode-and-forward scheme \cite{Cover79-TIT} and operates in a half-duplex manner, where the throughput is conversely proportional to the number of hops \cite{Gupta00-TIT}.
The nodes transmit predefined jamming signals with data signal to degrade Eve's wiretap ability by reducing the SINR of the wiretap channel \cite{Dong10-TSP}.
Meanwhile, legitimate nodes recover the original data by cancelling the jamming signals with predefined patterns \cite{Goel08-TWC}.

The index set of Eves is defined as $\mathcal{M}$.
We assume that Eves are present at each network layer to wiretap legitimate communications.
The Eves are assumed to be non-colluding as the SAGSIN's extensive geographic dispersion of nodes makes cooperative wiretap channels logistically unfeasible and the limitations in memory and compute of mobile Eves hinder storing raw I/Q signals.
We assume that \textit{the location and channel model of Eves are unknown}, where Eves in each network layer are independently distributed following the homogeneous Poisson point processes (HPPP) with density $\lambda_1, \lambda_2, \lambda_3$, and $\lambda_4$ for the space, air, ground, and sea network layers, respectively.

We define the index set of nodes as $\mathcal{I}=\{0, 1,...,I\}$.
A total of $U$ users are served through the multi-hop relays, where the user index set is denoted as $\mathcal{U}=\{1,...,U\}$.
The binary variable $x_{(i,j)}\in\{0,1\}$ indicates whether there exists a hop between the $i$-th and $j$-th nodes.
The location of the $i$-th node and Eve $e$ are denoted as $\bm{p}_i$ and $\bm{p}_e$, respectively.
\hl{All notations and variables are summarized in Table~\ref{tab:Summary of Notations} on Appendix~\ref{Supple:Notations and Variables}.}

\subsection{Propagation Model and Secrecy Metrics}
\label{subsec:Propagation Model and Secrecy Metrics}
The instantaneous received SNR of the legitimate hop between node $i$ and $j$, if any, can be given as
\begin{align}
    \mathrm{SNR}^{\mathsf{s}}_{(i,j)} = \frac{\rho_iG_{(i,j)}|h^{\mathsf{s}}_{(i,j)}|^2}{n_0(d^{\mathsf{s}}_{(i,j)})^{\alpha_{i}}}
\end{align}
where \hl{$G_{(i,j)}$ is the antenna gain}, $d^\mathsf{s}_{(i,j)}=\|\bm{p}_i-\bm{p}_j\|_2$ is a distance between; $\rho_i$ is the power density, $h_{(i,j)}$ is small-scale fading, \hl{and $n_0$ is the noise spectral density}. 
The path loss exponent $\alpha_i$ is determined based on the network layer in which node $i$ is located.
As nodes can decode and cancel the predefined jamming signal, the SINR of legitimate links are not affected by jamming.
Meanwhile, the SNR of wiretap link for node $i$ and Eve $k$ can be defined as
\begin{align}
    \hlmath{\mathrm{SNR}^{\mathsf{e}}_{i, e} = \frac
    {\rho_iG_{(i,j)}|h^{\mathsf{e}}_{(i, e)}|^2(d^{\mathsf{e}}_{(i, e)})^{-\alpha_i}}
    {\sigma_iG_{(i,j)}|h^{\mathsf{e}}_{(i, e)}|^2(d^{\mathsf{e}}_{(i, e)})^{-\alpha_i}+n_0}}
\end{align}
where $\sigma_i$ is \hl{the power spectral density} of the transmitted jamming signal \cite{LV20-TCOM, Chen22-TVT, Mao24-WCL}.
\hl{We consider that Eves have the same receiving capability as the legitimate nodes, assuming the same gain for Eves.}
\hl{For all links, we consider that $|h^{\mathsf{s}}_{(i,j)}|^2$ and $|h^{\mathsf{e}}_{i,e}|^2$ are Rayleigh fading channel, following exponential distribution $\mathsf{Exp}(1)$.
Adopting the Rayleigh model allows us to analyze the secrecy performance under worst-case conditions, 
though the LoS nature of SAGSIN channels may be more accurately captured by Rician or Nakagami-$m$ fading. 
In this sense, the Rayleigh assumption provides a meaningful lower bound on the SPSC probability, as further elaborated in Appendix~\ref{supple:SPSC Probability with Various Fading Effects}.
}
The entire transmission power density of node $i$ is limited to \hl{$P_i^{\max}$}; and the minimum transmission power density is defined as \hl{$P_i^{\min}$}, denoted as 
\begin{align}
    \hlmath{
    \rho_i + \sigma_i \leq P_i^{\max},~\rho_i \geq P_i^{\min}~ \forall i\in\mathcal{I}
    }
\end{align}

As the multiple non-colluding Eves can wiretap the legitimate transmission, the secrecy capacity of the link between nodes $i$ and $j$ is defined as
\begin{align}
    \hspace{-.2cm}
    C_{(i,j)} = \Big[\log_2\big(1+\SNR{s}_{(i,j)}\big)-\log_2\big(1+\max_{e\in\mathcal{M}}\SNR{e}_{i,e}\big)\Big]^+.
    \label{eq:secrecy capacity}
\end{align}
Then, as introduced in \cite{Jameel19-SurvTut, Kong18-TVT, Tolossa18-systemsJ, Bhargav16-TCOM}, the SPSC probability between node $i$ and $j$ is defined as
\begin{align}
    \mathbb{P}_{(i,j)}=\mathbb{P}(C_{(i,j)}>0).
    \label{eq:SPSC_0}
\end{align}
\textbf{The closed-form derivation of \eqref{eq:SPSC_0} will be presented in \eqref{eq:SPSC_final_approx_2}.}

\subsection{Multi-hop Relay Model: Graph-Theoretic Viewpoint}

\paragraph{Graph topology} The routing problem is generally considered as finding a sub-graph of directed graph $\mathcal{G}_{\mathrm{all}}=(\mathcal{N}, \mathcal{E})$ where $\mathcal{N}\subset \mathcal{I}\cup\mathcal{U}$ and $\mathcal{E}\subset \{(i,j)|i,j\in\mathcal{N}\}$.
This represents that nodes and users are the graph nodes and the relay hops are the graph edges \cite{Pagin22-TWC, Yin22-TCOM}.

We assume a spanning tree (ST) topology, which refers to an acyclic graph where all nodes have a single path to the root node \cite{3gpp.38.401}.
This assumption can be relaxed into a directed acyclic graph (DAG) topology, another topology discussed in the 3GPP standards \cite{3gpp.38.340}, where nodes can be backhauled by multiple parent nodes.
However, the DAG topology requires a sophisticated channel and data management to synchronize and integrate multiple backhaul links, which is not suitable for SAGSINs where node distances can vary significantly.\footnote{Nonetheless, the system model can also be applied to the relay with directed acyclic graph topology by changing the graph topology constraints.}

We define $\mathcal{E}_u$ to denote a relay from node $0$ to user $u$, and
\begin{align}
    \hspace{-.2cm}
    x_{(i,j)} = \begin{cases}
    1, \text{ if } (i,j)\in\mathcal{E}\\
    0, \text{ otherwise}
    \end{cases}
    \hspace{-0.4cm},~
    x_{(i,j),u} = \begin{cases}
    1, \text{ if } (i,j)\in\mathcal{E}_u \\
    0, \text{ otherwise}
    \end{cases}
    \hspace{-0.4cm}
\end{align}
to indicate whether edge $(i,j)$ belongs to the edge set of graph $\mathcal{G}$ and edge $(i,j)$ belongs to $\mathcal{E}_u$.\footnote{The system model considers cross-layer optimization of both network-layer routing and physical-layer resource, so both $x_{(i,j)}$ and $x_{(i,j),u}$ are required to consider the routing and radio resource optimization, respectively.}
There must be no edges entering to the root node $0$ and every other node should have exactly one entering edge, which can be denoted as
\begin{align}
    \sum_{i\in \mathcal{N}} x_{(i,0)} = 0,~
    \sum_{j\in \mathcal{N}} x_{(j,i)} = 1, \forall i \in \mathcal{N}\setminus \{0\}
    \label{eq:constraint_ST_topology_1}
\end{align}

We introduce the \textit{cut-based} constraint, which indicates that graph $\mathcal{G}$ has ST topology. 
For any non-empty subset $S \subset \mathcal{N} \setminus \{0\}$, there must be at least one selected edge $(i,j)$ with $i \in \mathcal{N} \setminus S$ and $j \in S$, which is denoted as
\begin{align}
    \sum_{i\notin S,\; j\in S} x_{(i,j)} \ge 1,
    \quad 
    \forall\, S \subset \mathcal{N}\setminus\{0\},\; S \neq \emptyset.
    \label{eq:constraint_ST_topology_3}
\end{align}
This constraint prevents any group of nodes from being disconnected from the root node.
Consequently, each node is guaranteed a directed path originating from the root, while satisfying the ST topology.

\paragraph{Relay throughput} The nodes dynamically allocate bandwidth and transmit power on the backhaul link to maximize the network utility.
We define $\beta_{(i,j),u}$ to represent the allocated bandwidth of user $u$ at link $(i,j)$; and the spectral efficiency $\gamma_{(i,j)}$ between node $i$ and $j$ as\footnote{The ergodic spectral efficiency $\mathbb{E}_{|h^\mathsf{s}_{(i,j)}|^2}\big[\log_2(1+\SNR{s}_{(i,j)})\big]$ is approximated as $\log_2\big(1+{\rho_i}/[{n_0(d^\mathsf{s}_{(i,j)})^{\alpha_i}]}\big)$ 
by using the Jensen's inequality. Appendix~\ref{Supple:Approximation of Ergodic Spectral Efficiency} provides the analysis on the estimation error.}
\begin{align}
    \gamma_{(i,j)} = \log_2\big(1+\SNR{s}_{(i,j)}\big).
    \label{eq:spectral_efficiency}
\end{align}
The relay throughput of user $u$ is defined as the minimum link capacity divided by the number of hops, denoted as
\begin{align}
    \eta_u = \min_{(i,j)\in\mathcal{E}_u}\frac{\beta_{(i,j),u}\gamma_{(i,j)}}{h_u},
\end{align}
where $h_u = \sum_{(i,j)\in\mathcal{E}} x_{(i,j),u}$.
Then, the minimum throughput of the system is defined as $\min_{u\in\mathcal{U}} \eta_u$.

\section{Problem Formulation and Proposed Solution} \label{sec:problem_solution}
We aim to maximize the minimum throughput across the entire system, ensuring fair service for users accessing the core internet via multiple hops.
\hl{
The problem of guaranteeing a SPSC probability in multi-hop relay terrestrial networks has been studied in \cite{Yao18-CL, Yao18-TWC, Thapar20-TVT}.
Similarly, the relay routing problem that ensures users’ QoS while achieving the optimal secure connection probability has also been investigated \cite{Xu16-ICC, Xu17-WCNC}.
However, there has been no attempt to optimize the max-min throughput of the multi-hop relaying system, guaranteeing the SPSC probability under cooperative jamming in SAGSINs.
}

Let $\mathbf{B}=\{\beta_{(i,j),u}:i,j\in \mathcal{N}, u\in\mathcal{U}\}$, $\mathbf{P}=\{\rho_i:i\in\mathcal{N}\}$, and $\mathbf{J}=\{\sigma_i:i\in\mathcal{N}\}$.
The max-min throughput problem is formulated as
\begin{subequations}
    \begin{alignat}{3}
        & \bf{\mathdutchcal{P}1}: && \max_{\mathcal{G}, \mathbf{B}, \mathbf{P}, \mathbf{J}} && 
        \min_{
            \substack{u\in\mathcal{U}, \\ (i,j)\in\mathcal{E}_u}
        }
        \frac{\beta_{(i,j),u}\gamma_{(i,j)}}{h_u}
        \label{P1:objective} \\
        &  && \mathrm{~~~~s.t.~} && 
        \mathbb{P}_{(i,j)} \geq x_{(i,j)} \tau,
        \label{P1:SPSC_probability_threshold}
        \\
        & && && \eqref{eq:constraint_ST_topology_1},~ \eqref{eq:constraint_ST_topology_3}
        \label{P1:ST_topology}\\
        & && && \sum_{j\in\mathcal{N}}\sum_{u\in\mathcal{U}} \beta_{(i,j),u} \leq B, 
        \label{P1:resource_allocation}\\
        & && && \rho_i+\sigma_i \leq \hlmath{P_i^{\max}}, 
        \label{P1:power_jamming_sum}\\[-0.1cm]
        & && && \beta_{(i,j),u} \geq 0, \rho_i \geq \hlmath{P_i^{\min}}, \sigma_i \geq 0.
        \label{P1:positive_control_variable}
    \end{alignat}
\end{subequations}
Constraint \eqref{P1:SPSC_probability_threshold} represents the threshold of the secrecy connection probability; \eqref{P1:ST_topology} is related to the graph topology; \eqref{P1:resource_allocation} is the frequency resource constraint; \eqref{P1:power_jamming_sum} pertains to the transmission and jamming power constraint; and \eqref{P1:positive_control_variable} ensures the positivity of resource variables.

There are two main challenges in Problem $\Problem{1}$.
First, The formulated problem is a mixed-integer non-convex problem due to the ST topology assumption and max-min throughput objective.
Specifically, the tight coupling between graph and radio resource variables in the objective makes it challenging to search the optimal solution \cite{lyu2025-TWC}. 
Finding the exact form of the SPSC probability is another challenge.
Numerical computation of \eqref{P1:SPSC_probability_threshold} may result in significant computational overhead and restrict finding feasible solutions under tight compute budgets.
This paper tackles the challenges above by separating the graph optimization from RRM; and by deriving the exact form of the SPSC probability.

\subsection{Closed-form Derivation of a Single-Hop SPSC Probability}
\label{subsec:Closed-form Derivation of a Single-Hop SPSC Probability}
For a hop between node $i$ and $j$, the SPSC probability is defined from \eqref{eq:secrecy capacity} as 
\begin{align}
    \mathbb{P}_{(i,j)} &=
    \mathbb{P}
    \left(
        \log_2\Big(
            \frac{1+\SNR{s}_{(i,j)}}{1+\max_{e\in\mathcal{M}}\SNR{e}_{i,e}}
        \Big)>0
    \right) \\
    &=\mathbb{P}\big(\SNR{s}_{(i,j)} > \max_{e\in\mathcal{M}}\SNR{e}_{i,e}\big)
    \label{eq:SPSC_2} \\
    &=\mathbb{E}_{|h|^2, \mathcal{M}}
    \Big[
        \prod_{e\in\mathcal{M}}\mathbb{P}\big(\SNR{s}_{(i,j)} > \SNR{e}_{i,e}\big)
    \Big].
    \label{eq:SPSC_3}
\end{align}
Equation \eqref{eq:SPSC_3} is derived from the property of HPPP and \hl{$\mathbb{P}(C >\max \{\gamma_1,...,\gamma_n\} )=\mathcal{P}(C>\gamma_1, ..., C>\gamma_n)$.}

We define $\lambda_i\in\{\lambda_1, \lambda_2, \lambda_3, \lambda_4\}$ as the Eve density of the network layer that node $i$ belongs to.
The probability generating functional of HPPP \cite{Yao16-TCOM} results in $\mathbb{E}_{\mathcal{M}}[\prod_{e\in\mathcal{M}}f(\bm{p}_e)]=\exp \big(-\lambda_i\int_{\mathbb{R}^2}1-f(\bm{p}_e) d\bm{p}_e\big)$.
Then, $\mathbb{P}_{(i,j)}$ becomes
\hlframe{
\begin{align}
    &\!\!\mathbb{P}_{(i,j)} \hspace{-0.05cm}= \hspace{-0.05cm}\mathbb{E}_{|h|^2}\!
    \bigg[\!
        \exp\!\Big[\!
            -\!\lambda_i\hspace{-0.1cm}\int_{\mathbb{R}^2}\hspace{-0.15cm}\mathbb{P}\big(\SNR{s}_{(i,j)}<\SNR{e}_{(i,e)}\big) ~d\bm{p}_e
        \!\Big]
    \!\bigg] 
    \hspace{-.25cm} \nonumber
    \label{eq:E1_derivation_2} \\[-.1cm]
    &
    \!\!\! = \hspace{-0.05cm}\mathbb{E}_{|h|^2}\!
    \bigg[\!
        \exp\!\Big[\!
    \\[-.1cm]
    &
    \!\! -\!\lambda_i\hspace{-0.1cm}\int_{\mathbb{R}^2}\hspace{-0.15cm}\mathbb{P}
                \bigg(
                    \underbrace{
                    \frac{
                        |h^\mathsf{s}_{(i,j)}|^2 (d^\mathsf{e}_{(i,e)})^{\alpha_i} n_0 
                    }{
                        (d^\mathsf{s}_{(i,j)})^{\alpha_i} n_0 - \sigma_i G_{(i,j)}  |h^\mathsf{s}_{(i,j)}|^2
                    }
                    }_{\displaystyle \textcolor{RoyalBlue}{\textbf{(a)}}}
                    <  |h^\mathsf{e}_{(i,e)}|^2
                \bigg) d\bm{p}_e
        \!\Big]
    \!\bigg].
    \nonumber
\end{align}
From $|h^\mathsf{e}_{i,e}| \sim \mathsf{Exp}(1)$, \eqref{eq:E1_derivation_2} simplifies to
\begin{align}
    \label{eq:E1_derivation_3}
    \mathbb{P}_{(i,j)} 
    = \mathbb{E}_{|h|^2}\!
    \bigg[\!
        \exp\!\Big[\!
            -\lambda_i\!
            \underbrace{
            \int_{\mathbb{R}^2}\!\!
            \exp\big(
                - \textcolor{RoyalBlue}{\textbf{(a)}}
            \big) ~d\bm{p}_e
            }_{\displaystyle \textcolor{RoyalBlue}{\textbf{(b)}}}
        \!\Big]
    \!\bigg].
\end{align}
Direct integral of \eqref{eq:E1_derivation_3} results in
\begin{align}
\hspace{-0.2cm}\textcolor{RoyalBlue}{\textbf{(b)}} = \frac{2\pi}{\alpha_i}\Gamma(\frac{2}{\alpha_i})(d^\mathsf{s}_{(i,j)})^{2}
        \bigg(
            \frac{
                |h^\mathsf{s}_{(i,j)}|^2
            }{
                1  - \frac{\sigma_i G_{(i,j)}}{(d^\mathsf{s}_{(i,j)})^{\alpha_i} n_0} |h^\mathsf{s}_{(i,j)}|^2
            }
        \bigg)^{-\frac{2}{\alpha_i}}\hspace{-0.2cm}.
\end{align}
Using the Jensen's inequality, this can be furthur approximated as
\begin{align}
    \label{eq:SPSC_final_approx_2}
    & 
    \tilde{\mathbb{P}}_{(i,j)} =
    \mathbb{E}_{|h|^2}
    \Big[
        \exp\big[
            -\lambda_i \textcolor{RoyalBlue}{\textbf{(b)}}
        \big]
    \Big] =
    \\ &
    \exp
    \bigg[
        -\kappa_i
        \Big[
        \Gamma \big(1-\frac{2}{\alpha_i} \big) 
        -
        \frac{2 \sigma_i G_{(i,j)}}
                 {\alpha_i (d^\mathsf{s}_{(i,j)})^{\alpha_i} n_0}
        \Gamma \big(2-\frac{2}{\alpha_i} \big)
        \Big]
        (d^\mathsf{s}_{(i,j)})^2
    \bigg]
    \nonumber
\end{align}
for $\kappa_i=\lambda_i \frac{2\pi}{\alpha_i}\Gamma(\frac{2}{\alpha_i})$.
Equation \eqref{eq:SPSC_final_approx_2} diverges when $\alpha_i=2$, which implies that mathematically guaranteed secure connection in the free-space cannot be achieved.
When $\sigma_i=0$, \eqref{eq:SPSC_final_approx_2} reduces to closed-form expression from a previous study \cite{Yao16-TCOM}, which calculates a SPSC probability without considering jamming.
This confirms that the new derivation is both valid and extensible.

Then, the approximation gap is bounded as
\begin{align}
    \mathbb{P}_{(i,j)}-\tilde{\mathbb{P}}_{(i,j)} \leq \frac{\lambda_i^2\pi^3}{2C}
    \Big(\ln R - \frac{\pi}{4}\Big)
    \label{eq:approximation gap}
\end{align}
for $C=    \big((d^\mathsf{s}_{(i,j)})^{\alpha_i} n_0 - \sigma_i G_{(i,j)}  |h^\mathsf{s}_{(i,j)}|^2 \big)^{-1}
$ \vspace{0.05cm}
and sufficiently large maximum eavesdropping range $R$.

A detailed derivation of \eqref{eq:SPSC_final_approx_2} is presented in Appendix~\ref{Supple:Derivation of the SPSC Probability}, 
and the analysis of the approximation gap \eqref{eq:approximation gap} is provided in Appendix~\ref{Supple:Derivation of the Bound Gap in the SPSC Approximation}.
}

\subsection{Optimal Frequency Resource Allocation}
We first look into how the optimal $\mathbf{B}$ is obtained for given $\mathcal{G}, \mathbf{P}, \mathbf{J}$.
\hl{Adopting auxiliary variable $\eta$, $\Problem{1}$ with fixed $\mathcal{G}, \mathbf{P}, \mathbf{J}$ can be considered as}
\begin{subequations}
    \begin{alignat}{3}
        & \Problem{2}: && \min_{\mathbf{B}, \eta} && 
        \frac{1}{\eta}
        \\
        &  && \mathrm{~~s.t.~} && 
        \beta_{(i,j),u}\gamma_{(i,j)} h_u^{-1} \geq x_{(i,j),u} \eta 
        \label{p2:constarint_1}
        \\
        & && && \sum_{j\in\mathcal{N}}\sum_{u\in\mathcal{U}} \beta_{(i,j),u} \leq B.
        \label{p2:constarint_2}
    \end{alignat}
\end{subequations}
The Lagrangian function of $\Problem{2}$ is
\begin{align}
    \mathcal{L}(\mathbf{B}, \eta, \bm{\lambda}, \mu) =
    \frac{1}{\eta} &+ 
    \sum_{\lambda_{(i,j),u}\in\bm{\lambda}} \lambda_{(i,j),u}\Big(
    x_{(i,j),u} \eta -
    \frac
        {\beta_{(i,j),u}\gamma_{(i,j)}}
        {h_u} 
    \Big)
    \nonumber \\
    &+ \mu \big(\sum_{j\in\mathcal{N}}\sum_{u\in\mathcal{U}} \beta_{(i,j),u}-B \big)
    \label{eq:Lagrange_beta}
\end{align}
for $\bm{\lambda}=\{\lambda_{(i,j),u}:i,j\in \mathcal{N}, u\in\mathcal{U}\}$.
The first-order optimality condition on \eqref{eq:Lagrange_beta}, $\delta\mathcal{L}/\delta\eta=0, ~ \delta\mathcal{L}/\delta\beta_{(i,j),u}=0$, gives
\begin{align}
    \eta = \big(\sum \lambda_{(i,j),u} x_{(i,j),u} \big)^{-\frac{1}{2}},~\lambda_{(i,j),u} = 
    \frac
        {\mu h_u} 
        {\gamma_{(i,j)}}.
    \label{eq:Lagrange_beta_2}
\end{align}
The optimal resource allocation $\beta^*_{(i,j),u}$ occurs when all frequency resources are fully utilized, which corresponds to the equality condition of \eqref{p2:constarint_2}. 
Otherwise, we have $\mu=\lambda_{(i,j),u}=0$ and $\eta$ in \eqref{eq:Lagrange_beta_2} is not defined according to complementary slackness of the KKT conditions.
Similarly, the equality condition of \eqref{p2:constarint_1} should be satisfied as well.

This can be expressed mathematically as 
\begin{align}
    x_{(i,j),u} \eta -
    \frac
        {\beta^*_{(i,j),u}\gamma_{(i,j)}}
        {h_u} = 0, ~
    \sum_{j\in\mathcal{N}}\sum_{u\in\mathcal{U}} 
    \beta^*_{(i,j),u}=B.
\end{align}
Solving the above equations with respect to $\beta^*_{(i,j),u}$ results in
\begin{align}
    \beta^*_{(i,j),u} =
        B\frac
        {x_{(i,j),u} h_u}
        {\gamma_{(i,j)}}
        \bigg(
            \sum_{j\in\mathcal{N}}\sum_{u\in\mathcal{U}}
            \frac
            {x_{(i,j),u} h_u}
            {\gamma_{(i,j)}}
        \bigg)^{-1}.
    \label{eq:optimal_beta}
\end{align}

\subsection{Optimal Power Allocation}
For fixed $\mathcal{E}$, we can reformulate $\Problem{1}$ by plugging in \eqref{eq:optimal_beta} and \hl{introducing an auxiliary variable $\eta$ as}
\begin{subequations}
    \begin{alignat}{3}
        & 
        \Problem{3}: && \min_{\mathbf{P}, \mathbf{J}, \eta} && 
        \frac{1}{\eta}
        \\
        &  && \mathrm{~~~s.t.~~} && 
        \eta 
            \sum_{j\in\mathcal{N}}\sum_{u\in\mathcal{U}}
            \frac
            {x_{(i,j),u} h_u}
            {\gamma_{(i,j)}}
        \leq B,
        \\
        & && &&
        \eqref{P1:SPSC_probability_threshold},~
        \eqref{P1:power_jamming_sum},~
        \rho_i \geq \hlmath{P_i^{\min}},~\sigma_i \geq 0.
    \end{alignat}
\end{subequations}

\hl{We rigorously prove the following argument by exploring KKT conditions: \textit{The optimal throughput is attained at the maximum transmission power allowed by the system}.
We first simplify the constraints and then demonstrate that the maximum achievable transmission power is determined by the minimum jamming power density.}

As $d^\mathsf{s}_{(i,j)}$ are fixed for given $\mathcal{E}$, the SPSC probability \eqref{eq:SPSC_final_approx_2} can be viewed as a function of $\sigma_i$, denoted as $\mathbb{P}_{(i,j)}(\sigma_i)$.
As $\mathbb{P}_{(i,j)}(\sigma)$ is a monotonically increasing function, the constraint \eqref{P1:SPSC_probability_threshold} can be redefined as
\begin{align}
    \hlmath{
    \sigma_i \geq \mathbb{P}^{-1}_{(i,j)}(\tau)
    =\frac{\alpha_i(d^\mathsf{s}_{(i,j)})^{\alpha_i}n_0}{2G_{(i,j)}\big(1-\tfrac{2}{\alpha_i}\big)}
    \bigg[
    1+\frac{\alpha_i\sin\big(\tfrac{2\pi}{\alpha_i}\big)}
    {2\pi^2\lambda_i(d^\mathsf{s}_{(i,j)})^{2}}\ln\tau
    \bigg].
    }
\end{align}
We remark that $\sigma_i$ is bounded by the maximum of the $\mathbb{P}^{-1}_{(i,j)}(\tau)$, which is  determined by the farthest link from node $i$.
Combined with $\sigma_i\geq 0$, the constraint can be given by
\begin{align}    
    \hlmath{
    \sigma_i \geq
    \tau_i =
    \frac{\alpha_i(d^\mathsf{max}_{i})^{\alpha_i}n_0}{2G_{(i,j)}\big(1-\tfrac{2}{\alpha_i}\big)}
    \bigg[
    1+\frac{\alpha_i\sin\big(\tfrac{2\pi}{\alpha_i}\big)}
    {2\pi^2\lambda_i(d^\mathsf{max}_{i})^{2}}\ln\tau
    \bigg]
    }
    \label{eq:definition_of_tau_i}
\end{align}
where $d^\mathsf{max}_{i}=\max_{j:(i,j)\in\mathcal{E}} x_{(i,j)}d^\mathsf{s}_{(i,j)}$. \vspace{.05cm}

Then, Problem $\Problem{3}$ can be simplified as  
\begin{subequations}
    \begin{alignat}{3}
        & 
        \Problem{3}: && \max_{\mathbf{P}, \eta} && 
        \eta
        \\
        &  && \mathrm{~~~s.t.~~} && 
        \eta             \sum_{j\in\mathcal{N}}\sum_{u\in\mathcal{U}}
            \frac
            {x_{(i,j),u} h_u}
            {\gamma_{(i,j)}}
        \leq B,
        \\
        & && &&
        \rho_i+\sigma_i \leq \hlmath{P_i^{\max}},~
        \rho_i \geq \hlmath{P_i^{\min}},~\sigma_i \geq \tau_i.
        \label{p3:power_jamming_sum}
    \end{alignat}
\end{subequations}

The Lagrangian function of $\Problem{3}$ is
\begin{align}
    \mathcal{L}(\mathbf{P}, \eta, \bm{\lambda}&, \bm{\mu}, \bm{\nu}, \bm{\xi}) = \eta
    - \sum_{i\in\mathcal{I}} \lambda_i
        \bigg[
            B - \eta 
            \sum_{j\in\mathcal{N}}\sum_{u\in\mathcal{U}}
            \frac
            {x_{(i,j),u} h_u}
            {\gamma_{(i,j)}}
        \bigg]
    \nonumber \\ &
    + \sum_{i\in\mathcal{I}} \mu_i(\hlmath{P_i^{\max}}-\rho_i-\sigma_i) + \sum_{i\in\mathcal{I}} \nu_i(\rho_i-\hlmath{P_i^{\min}}) \nonumber
    \\ &
    + \sum_{i\in\mathcal{I}}\xi_i(\sigma_i-\tau_i).
\end{align}
Then, the derivatives of Lagrangian function are provided as
\begin{align}
    \frac{\delta \mathcal{L}}{\delta \eta}
    &= 1 + \sum_{i\in\mathcal{I}}\lambda_i            \sum_{j\in\mathcal{N}}\sum_{u\in\mathcal{U}}
            \frac
            {x_{(i,j),u} h_u}
            {\gamma_{(i,j)}}
    = 0,
    \\
    \frac{\delta \mathcal{L}}{\delta \rho_i}
    &= \lambda_i \eta\frac{\delta A_i}{\delta \rho_i} - \mu_i + \nu_i = 0,
    \frac{\delta \mathcal{L}}{\delta \sigma_i}
    = -\mu_i + \xi_i = 0.
\end{align}
Also, the complementary slackness condition indicates $\nu_i=0$ and $\sigma_i>0$ 
since the nodes in the network can transmit signal to another node.

Suppose $\sigma_i > \tau_i$ for all $i\in\mathcal{I}$, meaning that we allocate more jamming power than the threshold.
Then, we have $\xi_i=\mu_i=0$ by the complementary slackness condition and $\frac{\delta \mathcal{L}}{\delta \sigma_i}$.
Plugging in $\mu_i=\nu_i=0$ on $\frac{\delta \mathcal{L}}{\delta \rho_i}$ results in $\lambda_i=0$, which contradicts $\frac{\delta \mathcal{L}}{\delta \eta}$.
Thus, the optimal point should have $\sigma_i=\tau_i$.

Constraint \eqref{p3:power_jamming_sum} now can be viewed as $\rho_i\leq \hlmath{P_i^{\max}} - \tau_i$.
If we assume $\rho_i<\hlmath{P_i^{\max}}-\tau_i$ for $i\in\mathcal{I}$, then $\mu_i=0$ to satisfy the complementary slackness condition.
Again, we obtain $\lambda_i=0$ to meet the optimality condition $\frac{\delta \mathcal{L}}{\delta \rho_i}$, which contradicts $\frac{\delta \mathcal{L}}{\delta \eta}$.

Combining the above conditions, the optimal transmission and jamming power, $\rho^*_i$ and $\sigma^*_i$, are given as
\begin{align}
    \rho^*_i = \hlmath{P_i^{\max}}-\tau_i,~\sigma^*_i = \tau_i,~\forall i \in \mathcal{I},
    \label{eq:optimal_rho_sigma}
\end{align}
for $\tau_i$ in \eqref{eq:definition_of_tau_i}.
We correspondingly denote the optimal spectral efficiency as $\gamma^*_{(i,j)}$ by putting $\rho^*_i$ into \eqref{eq:spectral_efficiency}.

The following theorem guarantees that the RRM in \eqref{eq:optimal_beta} and \eqref{eq:optimal_rho_sigma} is optimal for the given graph $\mathcal{G}$:
\begin{theorem}[\bf{Global optimality of RRM}]
Let $\mathcal{G}$ be a fixed network topology. Then, the solution $(\beta^*_{(i,j),u}, \rho^*_i, \sigma^*_i)$\vspace{.05cm} obtained by solving $\Problem{1}$ under the given $\mathcal{G}$ is globally optimal.
\label{theorem:1}
\end{theorem}
\begin{proof}
The objective function $1/\eta$ in $\Problem{2}$ is convex for $\eta>0$, and the constraints \eqref{p2:constarint_1} and \eqref{p2:constarint_2} are affine. 
Thus, the KKT conditions serve as both necessary and sufficient conditions, guaranteeing $\beta^*_{(i,j),u}$ is a global optimum solution of $\Problem{2}$. Similarly, in $\Problem{3}$, the objective function is convex, and all constraints including \eqref{p3:power_jamming_sum} are affine. Thus, KKT conditions are necessary and sufficient, and the solution $(\rho^*_i,\sigma^*_i)$ is globally optimal for $\Problem{3}$.

Then, from \eqref{eq:optimal_beta}, $\beta^*_{(i,j),u}$ can be expressed as an explicit function of $(\rho^*_i,\sigma^*_i)$. 
Therefore, combining these yields the full set of decision variables for $\Problem{1}$ under fixed $\mathcal{G}$ in the max-min throughput problem \cite{Boyd04-cvx, Gong11-TSP}.
\end{proof}

\subsection{Monte-Carlo Relay Routing}\label{sec:mc_routing}

\paragraph{Problem formulation} We have analytically found the optimal $\beta^*_{(i,j),u}$, $\rho^*_i$, and $\sigma^*_i$ for given graph $\mathcal{E}$.
Then, the problem can be reformulated as
\begin{align}
    \hspace{-.2cm}
    \max_{\mathcal{G}} \hspace{-.1cm}
        \min_{\substack{u\in\mathcal{U}, \\(i,j)\in\mathcal{E}_u}}
        \hspace{-.1cm}
        B 
        \bigg(
        \sum_{j\in\mathcal{N}}\sum_{u\in\mathcal{U}}
        \frac
        {x_{(i,j),u} h_u}
        {\gamma^*_{(i,j)}}
        \bigg)^{-1}
        \mathrm{s.t.}~ \eqref{P1:SPSC_probability_threshold},\eqref{P1:ST_topology}.
        \label{eq:routing_problem}
\end{align}

The SPSC constraint \eqref{P1:SPSC_probability_threshold} can be converted into an explicit link-distance constraint in the routing graph.
This transformation is necessary because SPSC probability does not directly indicate how the graph topology is constrained.
Specifically, the SPSC constraint determines the maximum link distance, defined as:
\begin{equation}
    d^{\mathsf{s}}_{(i,j)}\le D_i^{\mathsf{max}},\quad \forall j\in\mathcal{N}
    \label{eq:maximum_link_distance}
\end{equation}
where $D^{\mathsf{max}}_i$ indicates the farthest allowable distance from node $i$ at the maximum jamming power density \hl{$P_i^{\max} - P_i^{\min}$.}
\hl{As the SPSC approximation in \eqref{eq:SPSC_final_approx_2} increases monotonically with respect to $d^{\mathsf{s}}_{(i,j)}$, the value of $D^{\mathsf{max}}_i$can be determined by applying the bisection method to \eqref{eq:SPSC_final_approx_2} for the given $\tau$ and $\sigma_i=P_i^{\max} - P_i^{\min}$.}

\begin{figure}[htb]
    \centering
    \includegraphics[width=\linewidth]{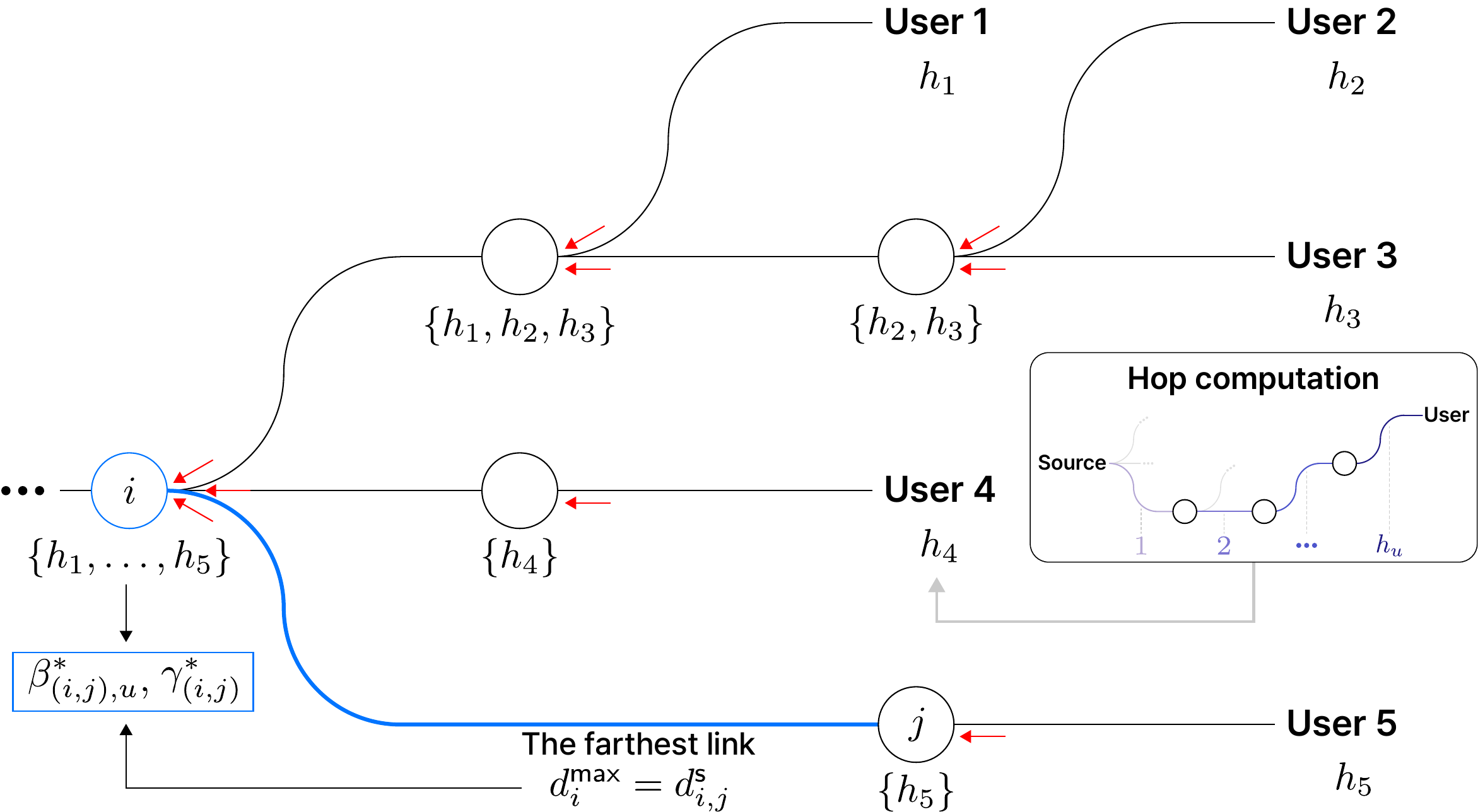}
    \caption{Visualization of variable flows for node $i$ to optimize radio resources.}
    \label{fig:graph_example}
\end{figure}

Solving the optimization problem in \eqref{eq:routing_problem} is challenging, as the objective cannot be computed before the graph is completed.
As depicted in Fig.~\ref{fig:graph_example}, determining $\beta^*_{(i,j),u}$, $\rho^*_i$, and $\sigma^*_i$ requires $h_u$ and $x_{(i,j),u}$ from the child nodes. 
These variables can be computed only after the complete routing graph $\mathcal{G}$ is established.
This interdependence between routing and resource allocation significantly complicates the routing optimization.

\paragraph{Routing algorithm} This challenge motivates the Monte-Carlo relay routing (MCRR) in Alg.~\ref{alg:MCRR}, which builds the relay graph $\mathcal{G}$ by sequentially adding source-to-user paths.
\hl{When a new user path is added to the graph, the MCRR algorithm propagates the user information from the user node back to the source node, as illustrated in Fig.~\ref{fig:graph_example}.
MCRR then exclusively updates the throughput for nodes along this path.
Nodes outside this path remain unaffected and therefore do not require re-optimization, making MCRR both node-specific and computationally efficient.
Leveraging the low computational complexity of the throughput calculation, MCRR iteratively swaps candidate user paths to construct a graph that yields the highest throughput.
We provide a visualization of the incremental update in Appendix~\ref{Supple:Visualization of the MCRR Algorithm}.}

The MCRR in Alg.~\ref{alg:MCRR} operates as follows:  
We first build the initial directed graph $\mathcal{G}_{\mathrm{all}}$ by including all possible edge whose length does not exceed the maximum feasible link distance $D_i$ in Eq.~\eqref{eq:maximum_link_distance} (\textbf{Line~1-2}).
Then, for each user $u$, we generate a set of $K$ candidate source-to-user path by collecting a shortest-path on $\mathcal{G}_{\mathrm{all}}$, assigning independent random weights $\mathcal{U}(0,1)$ to edge costs at each iteration (\textbf{Line~3-8}).
This mechanism thus operates similar to a biased random walk from node $0$ toward user $u$, constructing paths with an proper hop count while exploring topologies that may offer superior performance.

The MCRR refinement process (\textbf{Line~9-20}) creates multiple candidate graphs by removing and replacing paths to select the optimal configuration. 
It begins by temporarily removing the existing path of the user evaluated from the current routing graph (\textbf{Line~11-12}).
Then, each pre-computed candidate path is inserted once at a time, creating different graph configurations that satisfies ST topology \eqref{P1:ST_topology} (\textbf{Line~14}). 
As each change is implemented, its effects immediately propagate through the bandwidth and power allocation equations to upstream nodes (\textbf{Line~15}). 
The algorithm evaluates the minimum throughput of each candidate graph, committing the changes only if the throughput improves; otherwise, it reverts to the original path (\textbf{Line~17-18}).

\begin{algorithm}[htb]
\caption{Monte-Carlo Relay Routing (MCRR)}\label{alg:MCRR}
\begin{algorithmic}[1]
\REQUIRE Nodes $\mathcal{N}$ 
\ENSURE Solution graph $\mathcal{G}^*=(\mathcal{N}^*,\mathcal{E}^*)$

\STATE $\mathcal{E} \leftarrow \{(i,j)|d^\mathsf{s}_{(i,j)}\leq D_i,\forall i,j\in\mathcal{N}\}$,
$\mathcal{N}^*\leftarrow \{\}$, $\mathcal{E}^*\leftarrow \{\}$\!\!\!\!\!
\STATE $\mathcal{G}_{\mathrm{all}} \leftarrow (\mathcal{N}, \mathcal{E})$, $\mathcal{G}^*\leftarrow (\mathcal{N}^*, \mathcal{E}^*)$, $T_{\mathcal{G}^*}\leftarrow 0$

\FOR{user $u$ in $\mathcal{U}$}
    \FOR{$k$ in $1,\dots,K$}
        \STATE $\!\mathcal{E}_{u,k}\!\leftarrow\!$ Biased random walk from node $0$ to $u$ in $\mathcal{G}_{\mathrm{all}}$\!\!\!
        \STATE $\!\mathcal{N}_{u,k}\leftarrow$ nodes in the path $\mathcal{P}_{u,k}$
    \ENDFOR
\ENDFOR
\WHILE{$T-T_{\mathcal{G}^*}<\epsilon$}
    \FOR{each user $u$}  
        \STATE $\mathcal{E}_{-u} \leftarrow \{(i,j)\in\mathcal{E}^*| \sum_{u\in\mathcal{U}}x_{(i,j),u}>1\}$
        \STATE $\mathcal{N}_{-u} \leftarrow \{(i,j)~|~(i,j)\in\mathcal{E}_{-u}\}$
        \FOR{$k$ in $1,..,K$}
            \STATE $\mathcal{G}_{u,k}\leftarrow(\mathcal{N}_{-u}\cup\mathcal{N}_{u,k},~ \mathcal{E}_{-u}\cup\mathcal{E}_{u,k})$
            \STATE $\displaystyle T_{\mathcal{G}_{u,k}} \leftarrow \min_{i \in \mathcal{N}_{u,k}} \hspace{-.1cm} B \Big( \sum_{j\in\mathcal{N}}\sum_{u\in\mathcal{U}} \frac {x_{(i,j),u} h_u} {\gamma^*_{(i,j)}} \Big)^{-1}$
        \vspace{-.15cm}
        \ENDFOR
        \STATE $\mathcal{G}^*\leftarrow \!\arg\!\min_{\mathcal{G}\in \{\mathcal{G}^*,\mathcal{G}_{u,1},\dots,\mathcal{G}_{u,k}\} } T_\mathcal{G}$, $(\mathcal{N}^*, \mathcal{E}^*) \leftarrow \mathcal{G}^*$\!\!\!
        \STATE $T\leftarrow \min_{\mathcal{G}\in \{\mathcal{G}^*,\mathcal{G}_{u,1},\dots,\mathcal{G}_{u,k}\} } T_\mathcal{G}$
    \ENDFOR
\ENDWHILE
\end{algorithmic}
\end{algorithm}

\paragraph{Computational complexity of MCRR} Let $N$ be the number of nodes, $K$ the number of candidate paths per user, and $R$ the number of refinement rounds.
The computational cost of MCRR is decomposed as follows:
\begin{itemize}
    \item \textbf{Path sampling}: $K$ randomized shortest paths takes the complexity of $\mathcal{O}(KN\log N)$.

    \item \textbf{Iterative refinement}: In each of the $R$ rounds, every user examines $K$ candidate paths. For each, the algorithm updates the graph and recomputes throughput by locally propagating changes to relevant upstream nodes, giving total refinement cost $\mathcal{O}(RNK)$.
\end{itemize}
Then, total complexity of Alg.~\ref{alg:MCRR} is $\mathcal{O}(KN(\log N+R))\approx\mathcal{O}(KN\log N)$ when $N\gg R$.

\section{Numerical Experiments}\label{sec:experiments}
This section addresses the following research questions:

\noindent
\textbf{RQ 1.} How accurate is the closed-form approximation of the SPSC probability? \ding{220} Sec.~\ref{subsec:SPSC Probability Analysis}\\
\textbf{RQ 2.} How does the max–min throughput in SAGSIN change when security parameters change? \ding{220} Sec.~\ref{subsec:mmf_results} \\
\textbf{RQ 3.} How does the proposed framework perform in real-world network deployments? {\ding{220}~Sec.~\ref{subsec:demo}}

Answering the questions, we broadly provide a comprehensive validation from the feasibility of the system model to the effectiveness of the proposed scheme. 
Numerical experiments and analysis exploring each question are presented in Sec.~\ref{sec:experiments}-A to Sec~\ref{sec:experiments}-C. 
All simulations are implemented in Python~3.12 on an AMD Ryzen\textsuperscript{\texttrademark}~9~5800X processor.

\subsection{Analysis on SPSC Probability Approximation}\label{subsec:SPSC Probability Analysis}

\subsubsection{Calibration of the SPSC approximation over Eve density}
\label{sec:spsc_density}

\begin{figure}[htb]
  \centering
  \includegraphics[width=\columnwidth]{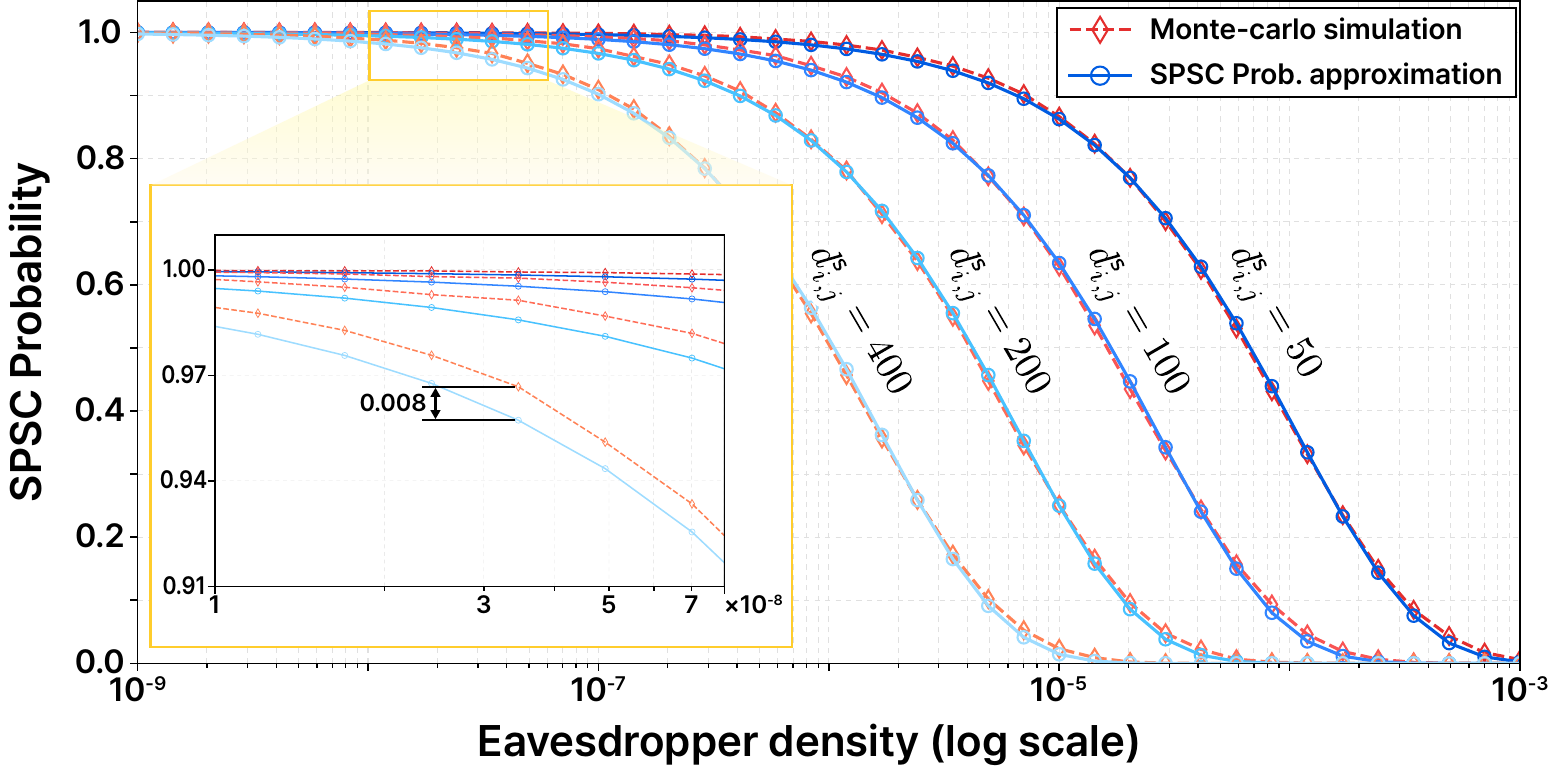}
  \vspace{-.4cm}
  \caption{\!The SPSC probability versus Eve density for various link distances.}
  \label{fig:spsc_density}
\end{figure}

Fig.~\ref{fig:spsc_density} shows comparisons of the 50,000 Monte–Carlo evaluations of \eqref{eq:SPSC_0} and its closed–form approximation in \eqref{eq:SPSC_final_approx_2} for four representative link lengths of $\{50,\,100,\,200,\,400\}\,\mathrm{km}$.

\hlframe{We introduced two empirical calibration parameters, $a_d$ and $p_d$, to compensate for the approximation gap.
These factors serve to correct complex, non-ideal effects in the actual environment that the theoretical model may fail to capture.
Accordingly, \eqref{eq:SPSC_final_approx_2} can be calibrated as
\begin{align}
    \label{eq:calibarated_SPSC_approximation}
    \hat{\mathbb{P}}_{(i,j)}&
    = \exp\bigg[
        -a_d \kappa_i^{p_d}
    \\[-0.2cm] & ~
        \times \Big[
        \Gamma \big(1-\frac{2}{\alpha_i} \big) 
        -
        \frac{2 \sigma_i G_{(i,j)}}
                 {\alpha_i (d^\mathsf{s}_{(i,j)})^{\alpha_i} n_0}
        \Gamma \big(2-\frac{2}{\alpha_i} \big)
        \Big]        \bigl(d_{(i,j)}^{\mathsf{s}}\bigr)^2
    \bigg].
    \nonumber
\end{align}
This correction preserves the inherent correlation between distance and SPSC probability while ensuring accurate estimation under varying Eve densities.} 
The \((a_d,p_d)\) pairs used in our experiments are:
$(a_{50},p_{50})=(0.224, 0.806)$;
$(a_{100},p_{100})=(0.170, 0.805)$;
$(a_{200},p_{200})=(0.133, 0.807)$; and
$(a_{400},p_{400})=(0.102, 0.807)$.

Across the whole distance range, the analytical expression faithfully tracks the Monte–Carlo result; the largest absolute gap occurring in the inset is $<0.008$, confirming that \eqref{eq:SPSC_final_approx_2} tightly approximates $\mathbb{P}_{(i,j)}$.
\hl{
An additional analysis on the various fading effects is provided in Appendix~\ref{supple:SPSC Probability with Various Fading Effects}.
}

\subsubsection{Accuracy of the SPSC approximation over link distance}
\label{sec:spsc_distance}

\begin{figure}[htb]
  \centering
  \includegraphics[width=\columnwidth]{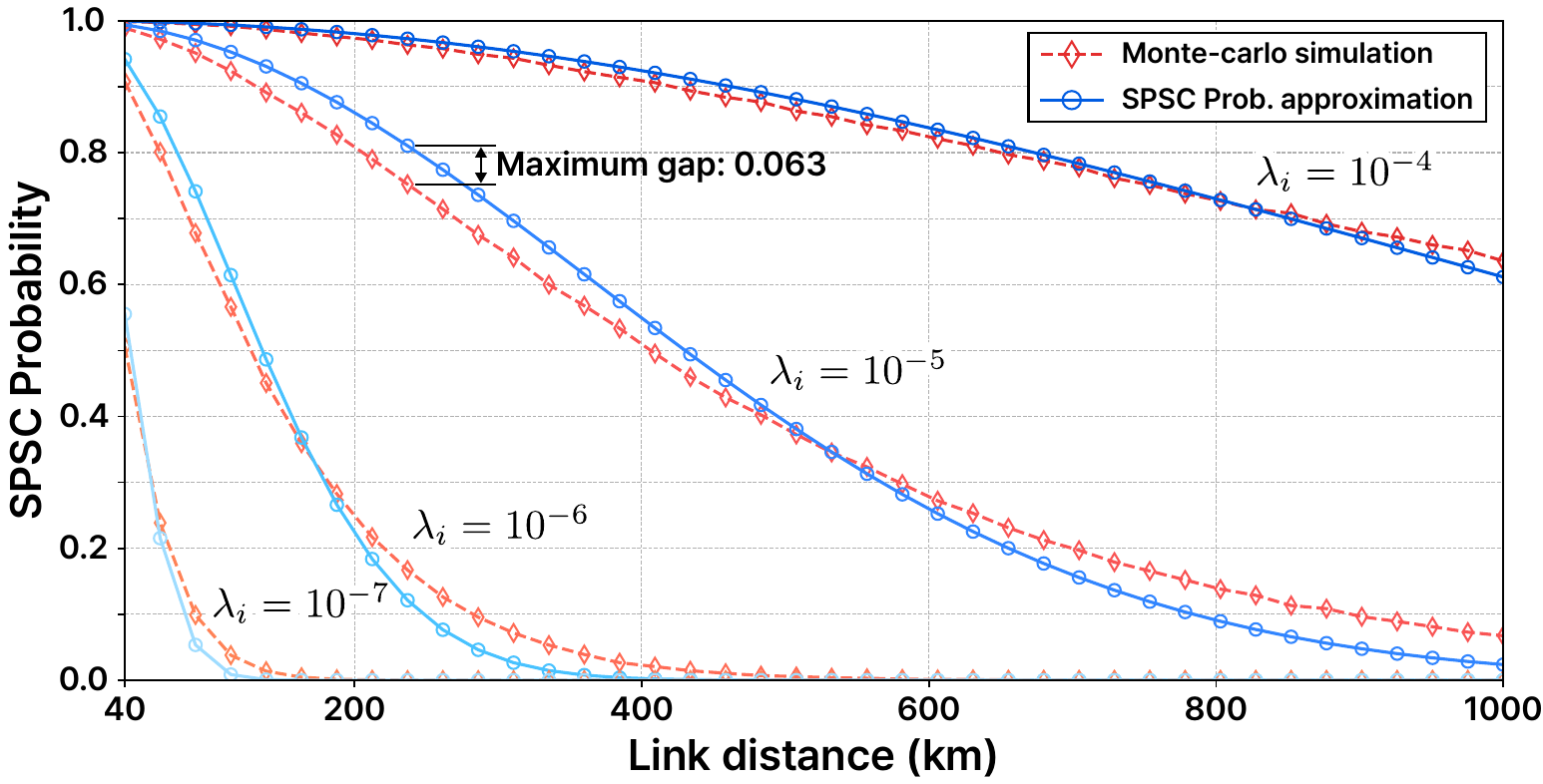}
  \vspace{-.4cm}
  \caption{\!The SPSC probability versus link distance for various Eve densities.}
  \label{fig:spsc_distance}
\end{figure}

Fig.~\ref{fig:spsc_distance} benchmarks the SPSC probability
obtained from 50,000 Monte–Carlo evaluations of \eqref{eq:SPSC_0} against the closed-form approximation \eqref{eq:SPSC_final_approx_2} for the four
Eve densities $\{10^{-4},10^{-5},10^{-6},10^{-7}\}$ km$^{-2}$.
The largest disparity, $0.063$, occurs around $320$ km of link distance when $\lambda_i=10^{-5}$.
When the link distance increases, the SPSC approximation shows larger degradation than the Monte-Carlo simulation.
This is attributed to the assumption that an infinite number of Eves are distributed over an infinite region, whereas in real environments both the number of Eves and the area are finite.

\subsubsection{Required jamming power versus link distance}
\label{sec:jamming_vs_d}

\begin{figure}[htb]
  \centering
  \includegraphics[width=\columnwidth]{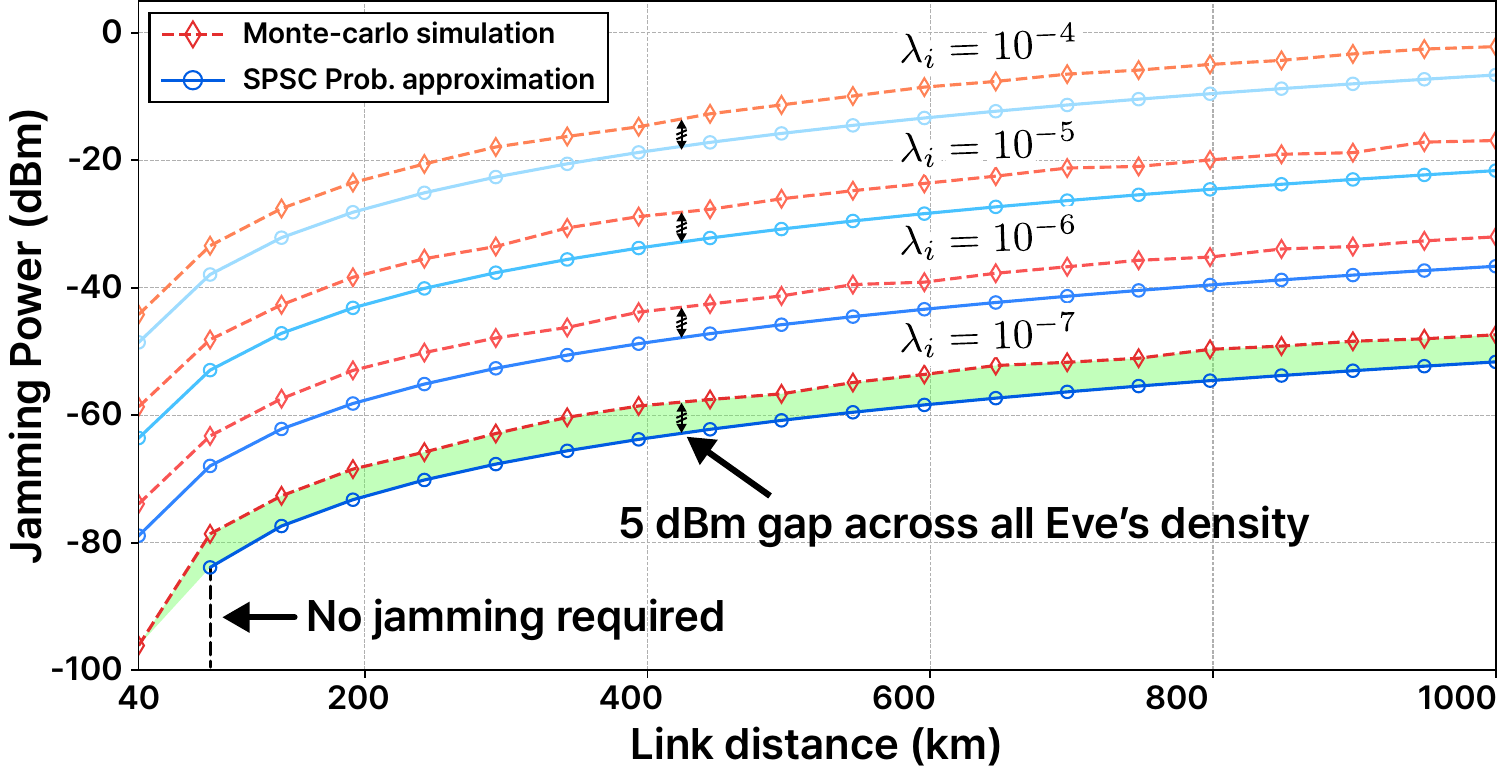}
  \vspace{-.4cm}
  \caption{\!Required jamming power versus link distance for various Eve densities.}
  \label{fig:jam_vs_distance}
\end{figure}

Fig.~\ref{fig:jam_vs_distance} plots the minimum jamming power to guarantee the target SPSC probability threshold $\tau=0.9999$.
The results are obtained (i) numerically from 500,000 Monte-Carlo simulation of \eqref{eq:SPSC_0} and (ii) via the closed-form inversion \eqref{eq:definition_of_tau_i} of the SPSC approximation.\footnote{The minimum jamming power is determined using the bisection method.
Due to the high variance of the experiments, the number of simulations is increased tenfold.}

As shown in Fig.~\ref{fig:jam_vs_distance}, the minimum required jamming power computed by the approximation consistently demonstrates a deviation of approximately 5 dBm across all regions.
This trend aligns with the over-estimation of the SPSC approximation in Fig.~\ref{fig:spsc_distance} when the link distance is relatively short. 
Figs.~\ref{fig:spsc_density} and~\ref{fig:jam_vs_distance} collectively suggest that the SPSC approximation \eqref{eq:SPSC_final_approx_2} needs to be corrected by the computed offset when determining the minimum jamming power in \eqref{eq:definition_of_tau_i} and the maximum link distance in \eqref{eq:maximum_link_distance}.

\subsection{Analysis on Max-Min Throughput }\label{subsec:mmf_results}

\begin{figure*}[htb]
  \centering
  \includegraphics[width=\textwidth]{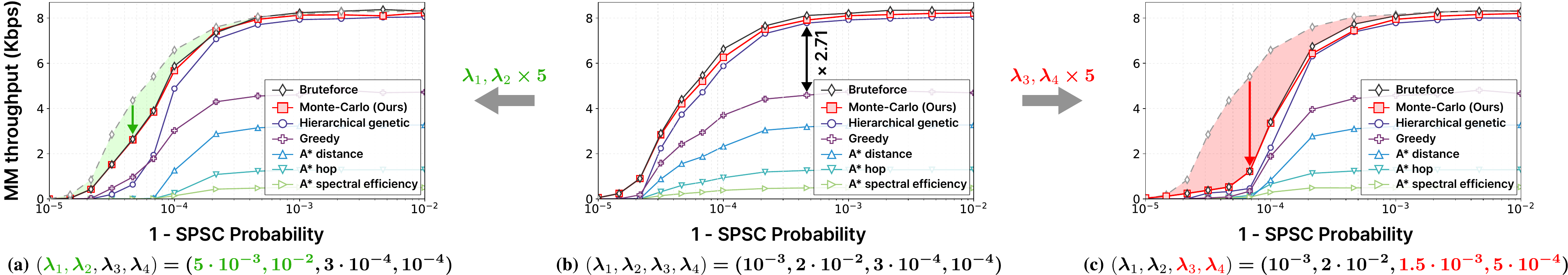}
  \vspace{-.4cm}
  \caption{Max-min throughput versus SPSC outage probability. 
  \hl{(a), (b), and (c) measure max-min throughput of SAGSINs in the same environmental settings, except for the Eve densities.}
  The gray dashed lines in (a) and (c) represent the numerical results from the Bruteforce method in (b), illustrating the relative throughput degradation as \textcolor[HTML]{44AC48}{\textbf{green}} and \textcolor[HTML]{EE2928}{\textbf{red}} areas in (a) and (b).
  (a) The max-min throughput of \textbf{Bruteforce} at $1-\tau=5{\cdot}10^{-5}$ decreases from $4.4$\,kbps in (b) to $2.6$\,kbps.
  (c) The performance drop is more pronounced, with the \textbf{Bruteforce} throughput at $1-\tau=7{\cdot}10^{-5}$ falling from $5.5$\,kbps to $1.2$\,kbps.
  }
  \label{fig:mmf_spsc}
\end{figure*}

We evaluate the max-min throughput by changing the security parameters and compare the proposed scheme with several baselines.
The SAGSIN networks are randomly generated using real-world terrain data across diverse latitudes and longitudes.

\paragraph{Baselines} Each baseline employs a different graph‐optimization strategy,
but they all adopt the same radio resource $\mathbf{B}$, $\mathbf{P}$, and $\mathbf{J}$ as the optimal radio‐resource allocation is obtained in Sec.~\ref{sec:problem_solution}.

\begin{itemize}
    \item \textbf{Bruteforce (Naive upper bound):} Exhaustively evaluates all routing graphs to find an upper-bound on minimum throughput.
    The number of exhaustive search trials is set to 5,000.
    \item \textbf{Hierarchical genetic:} Is a canonical genetic algorithm with elitism \cite{weise2009-genetic}, but operating in two hierarchical steps \cite{Schaefer2003-FOGA, Ciepiela2008-ICCS}.
    In the first step, each node is binary-encoded as a gene to determine the set $\mathcal{N}^*$.
    The second step determines the set of edges $\mathcal{E}$ by constructing feasible spanning trees among the selected node genes, subject to the graph topology constraints \eqref{P1:ST_topology} and the maximum link distance constraint \eqref{eq:maximum_link_distance}.
    The genetic algorithm is configured as 5,000 generations, 50 populations, 6 elite counts, and 5\% mutation probability.
    \item \textbf{Greedy:} Iteratively selects the relay path that maximizes immediate throughput gain of the user.
    \item \textbf{Variants of A$^*$:}
    Use fixed link cost metrics in A$^*$ search.
    While these metrics do not exactly optimize $\Problem{1}$ (because the throughput of nodes and edges changes depending on the graph topology), they serve as useful benchmarks for comparing the performance of different schemes.
    Moreover, graphs generated by A$^*$ always have spanning-tree structures \cite{Cormen2009-dijkstra}, ensuring that the graph solutions remain within the feasible domain of $\Problem{1}$.

    ~~The choice of metrics in A$^*$ is motivated by the objective function in $\Problem{1}$, which is\\
    \noindent - \textbf{A$^*$ distance:} Minimizes the summation of the link distances. \\
    - \textbf{A$^*$ hop:} Minimizes the number of source-to-user hops. \\
    - \textbf{A$^*$ spectral efficiency:} Maximizes the sum spectral efficiency of the source-to-user routes.
\end{itemize}

\begin{figure*}[htb]
  \centering
  \includegraphics[width=\textwidth]{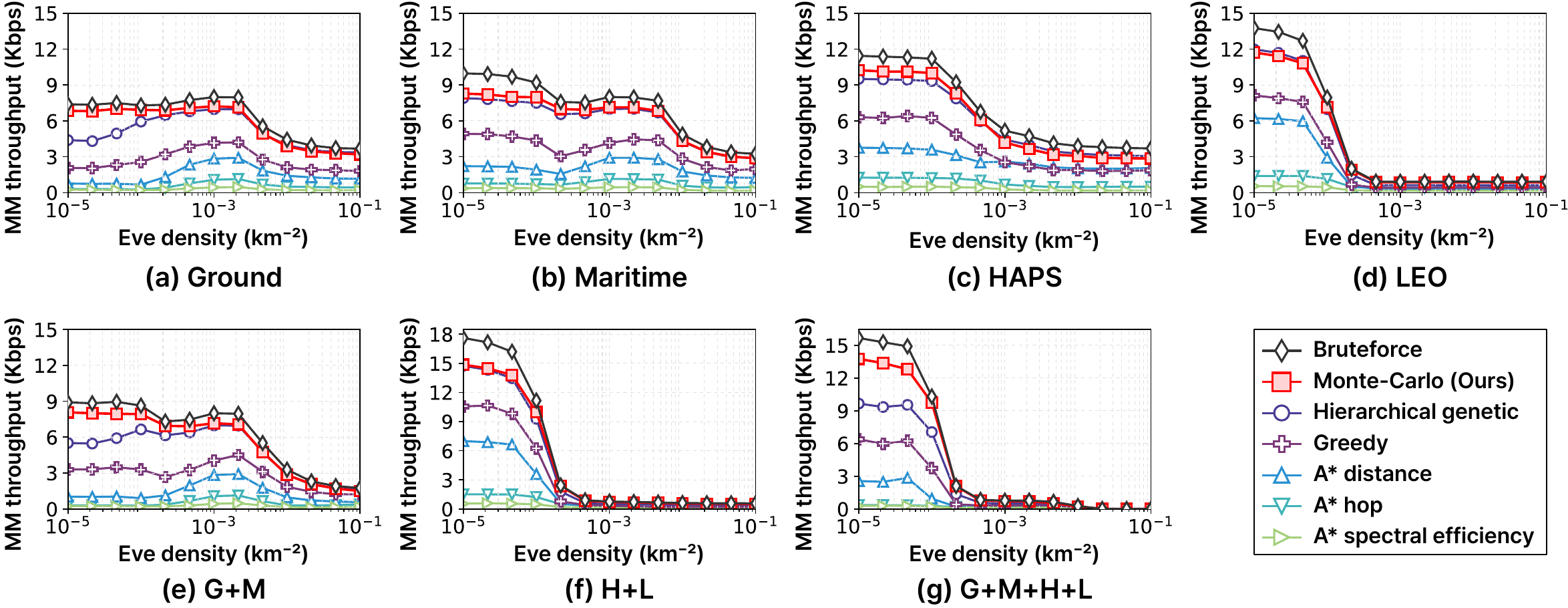}
  \vspace{-.5cm}
  \caption{Max-min throughput versus Eve density for various SAGSIN scenarios. 
  \hl{Each graph depicts the max–min throughput while varying the Eve density of the network layer indicated in the caption. The densities of the other layers are fixed to their respective values in the baseline configuration $(\lambda_1, \lambda_2, \lambda_3, \lambda_4) = (10^{-3}, 2{\cdot}10^{-3}, 3{\cdot}10^{-4}, 10^{-4})$.}
  Notations \textbf{G+M}, \textbf{H+L}, and \textbf{G+M+H+L} in (e), (f), and (g) correspond to ground and maritime; HAPs and LEO; and ground, maritime, HAPs and LEO, respectively.}
  \label{fig:mmf_density}
\end{figure*}

\begin{figure}[htb]
  \centering
  \includegraphics[width=\columnwidth]{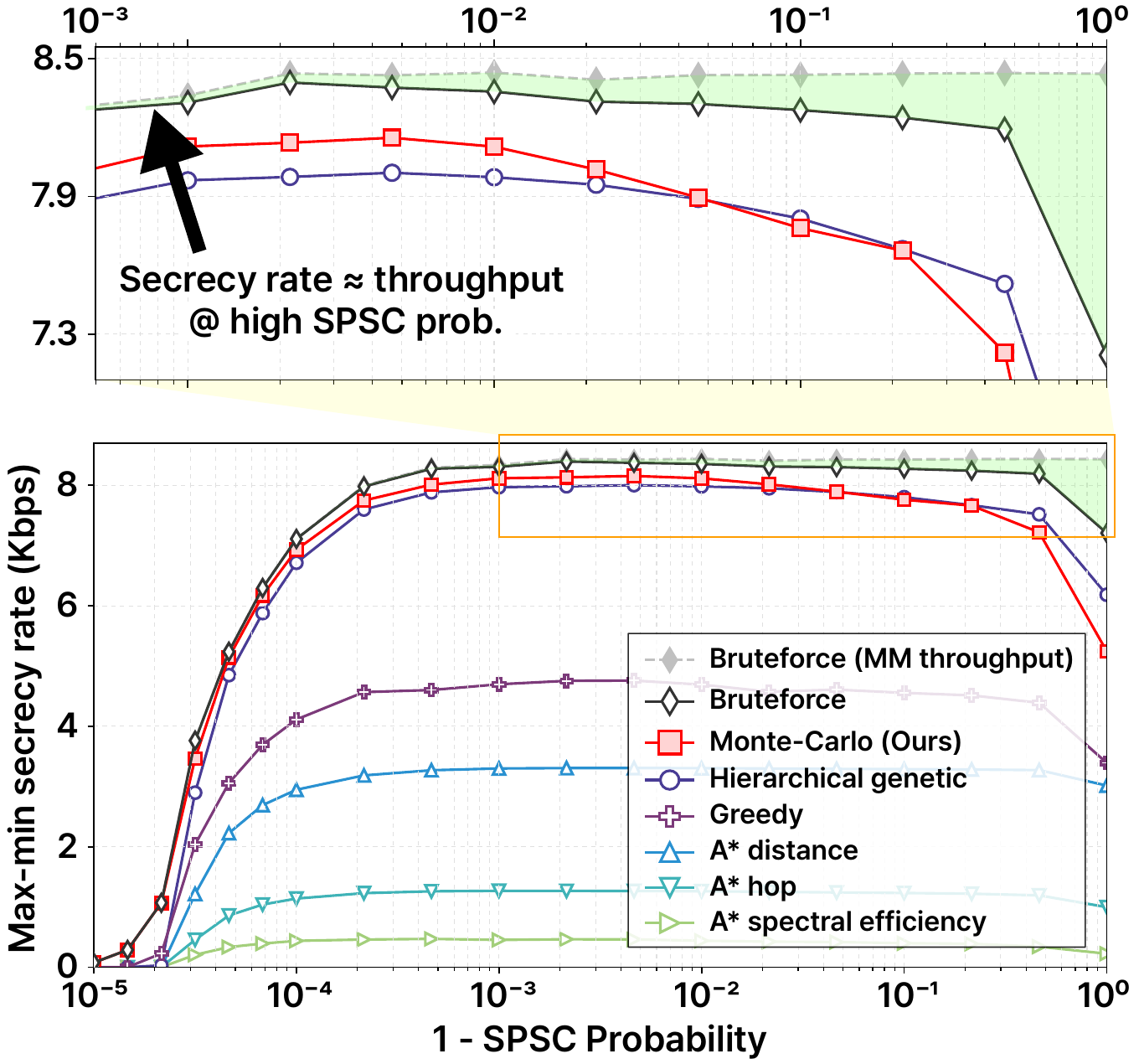}
  \vspace{-0.5cm}
  \caption{Comparison of max-min throughput and secrecy rate versus SPSC outage probability \hl{with the fixed Eve density of $(\lambda_1, \lambda_2, \lambda_3, \lambda_4) =
(10^{-3}, 2{\cdot}10^{-3}, 3{\cdot}10^{-4}, 10^{-4})$.}
  The lower graph plots the max-min secrecy rates across the SPSC probabilities, while the upper graph magnifies the gap (\textcolor[HTML]{44AC48}{\textbf{Green area}}) between the max-min throughput and secrecy rate for the Bruteforce solution.
  }
  \label{fig:mmsr_spsc}
\end{figure}

\begin{figure}[htb]
  \centering
  \includegraphics[width=\columnwidth]{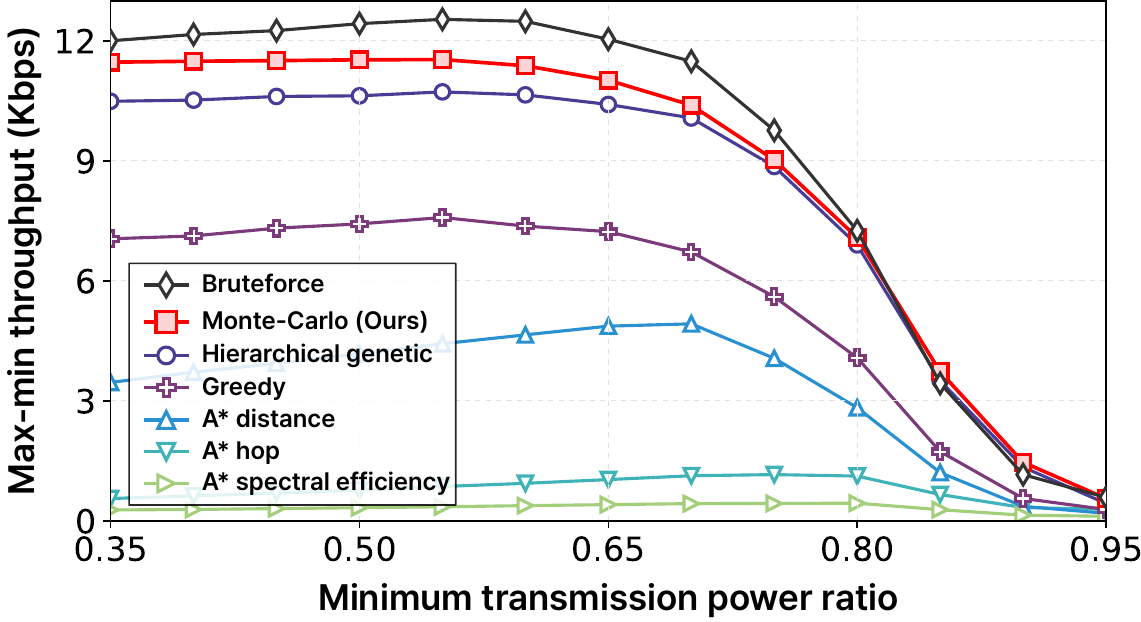}
  \vspace{-.4cm}
  \caption{Max–min throughput versus minimum transmit power ratio. \hl{The SPSC probability is set as $\tau=0.9999$ and the Eve density is fixed as $(\lambda_1, \lambda_2, \lambda_3, \lambda_4) =
(10^{-3}, 2{\cdot}10^{-3}, 3{\cdot}10^{-4}, 10^{-4})$.}}
  \label{fig:mmf_power}
\end{figure}

\paragraph{SAGSIN configurations} 
The SAGSIN map is generated considering realistic geography as illustrated in Sec.~\ref{subsec:demo},
randomly placing $150$ ground base stations, $150$ maritime base
stations, $12$ HAPs stations, and $10$ LEO satellites within
longitude $30^\circ$ and latitude $20^\circ$ to serve $60$ users.
The default Eve density is set as
$(\lambda_1, \lambda_2, \lambda_3, \lambda_4) =
(10^{-3}, 2{\cdot}10^{-3}, 3{\cdot}10^{-4}, 10^{-4})$
following the literature~\cite{Zou16-EveDensity}.
Physical layer parameters for SAGSIN, such as carrier frequency,
total bandwidth, antenna gain, and antenna gain-to-noise-temperature,
are configured as in Table~\ref{tab:simulation_parameters},
referring to~\cite{Liu22-Hanzo, Vondra18-IEEENetwork} and the 3GPP standard~\cite{3gpp.38.821}.

\begin{table}[htb]
    \centering
    \caption{Simulation Parameters}
    \label{tab:simulation_parameters}
    \begin{tabularx}{\if 1\doublecolumn 1 \else .62\fi \columnwidth}{@{\hspace{2.3mm}}c@{\hspace{2.3mm}}c@{\hspace{2.3mm}}c@{\hspace{2.3mm}}}
        \toprule
        Parameter (Unit)          & Ground (G), Maritime (M), HAPs (H)         & LEO               \\
        \cmidrule(l{0pt}r{2pt}){1-1} \cmidrule(l{2pt}r{2pt}){2-3}
        \arrayrulecolor{lightgray}
         Carrier frequency (GHz) & 14 & 20 \\
        \cmidrule(l{0pt}r{2pt}){1-1} \cmidrule(l{2pt}r{2pt}){2-3}
         Total bandwidth (MHz) & 250 & 400 \\
        \cmidrule(l{0pt}r{2pt}){1-1} \cmidrule(l{2pt}r{2pt}){2-3}
         Tx power (dBm) & 30 & 21.5 \\
        \cmidrule(l{0pt}r{2pt}){1-1} \cmidrule(l{2pt}r{2pt}){2-3}
         Tx antenna gain (dBi) &
         \begin{tabular}{@{}c@{}} 43.2 (G,M,H$\rightarrow$\hl{LEO}) \\  25 (G,M,H$\rightarrow$G,M,H)\end{tabular} &
         38.5 \\
        \cmidrule(l{0pt}r{2pt}){1-1} \cmidrule(l{2pt}r{2pt}){2-3}
         Rx antenna gain (dBi) &
         \begin{tabular}{@{}c@{}} 39.7 (G,M,H$\rightarrow$\hl{LEO}) \\ 25 (G,M,H$\rightarrow$G,M,H) \end{tabular} &
         38.5 \\
        \cmidrule(l{0pt}r{2pt}){1-1} \cmidrule(l{2pt}r{2pt}){2-3}
        \begin{tabular}{@{}c@{}} Antenna gain-to-noise \\ temperature (dB/K) \end{tabular} &
        \begin{tabular}{@{}c@{}} 1.5 (H$\rightarrow$LEO), 16.2 (H$\rightarrow$G,M,H) \\ 1.2 (G,M$\rightarrow$LEO), 15.9 (G,M$\rightarrow$G,M,H) \end{tabular} &
        13 \\
        \cmidrule(l{0pt}r{2pt}){1-1} \cmidrule(l{2pt}r{2pt}){2-3}
        Pathloss exponent & 2.8 (G), 2.7 (M), 2.6 (H) & 2.4 \\
        \arrayrulecolor{black}
        \bottomrule
    \end{tabularx}
\end{table}

\subsubsection{Impact of the SPSC probability threshold}

Fig.~\ref{fig:mmf_spsc} illustrates how the max-min throughput varies as the SPSC outage probability ($1-\tau$) increases from $10^{-5}$ to $10^{-2}$.
The minimum transmit power $P_i$ for each base station is set to $80\%$ of its total available power.
Figs.~\ref{fig:mmf_spsc}a and~\ref{fig:mmf_spsc}c are obtained by scaling the Eve densities of ground/maritime nodes; and HAPs/LEO nodes fivefold, relative to Fig.~\ref{fig:mmf_spsc}b.

The minimum throughput of all schemes converges to zero as $1-\tau$ approaches zero. 
The stricter SPSC threshold $\tau$ reduces the maximum link distance \eqref{eq:maximum_link_distance}, making some users unserviceable. 
Across all subfigures, the \textbf{Monte-Carlo} scheme consistently attains throughput levels within ${\approx}5\%$ of the optimal values from the \textbf{Bruteforce} method, demonstrating its efficiency and near-optimal performance.

In Figs.~\ref{fig:mmf_spsc}a and~\ref{fig:mmf_spsc}c, we note that the increase in the Eve density of HAPs and LEO notably deteriorate the max-min throughput, inducing a long tail near zero throughput. 
This result verifies the critical role of HAPs and LEO nodes in forming secure relays in SAGSINs, as they can establish long-distance connections more easily due to their low path loss and line-of-sight characteristics.


\subsubsection{Impact of Eve density}

Fig.~\ref{fig:mmf_density} shows how the max-min throughput changes when Eve densities in the various SAGSIN layers vary.
The SPSC probability threshold $\tau$ is fixed at $99.99\%$. 
In all subfigures, the max-min throughput improves as the Eve density decreases. 
This improvement is attributed to the expansion of the feasible region defined by the constraints in $\Problem{1}$, thereby providing a broader set of relay options to achieve higher user throughput.

Nevertheless, as shown in Figs.~\ref{fig:mmf_density}a, b, and e, the max-min throughputs of ground and maritime nodes exhibit only modest throughput improvement relative to those of HAPs and LEO nodes. Although the large number of ground and maritime nodes expands the search space and theoretically offers more routing possibilities, it also significantly increases algorithmic complexity and reduces the likelihood of reaching the global optimum. 
Consequently, the additional computational burden offsets the expected throughput gain due to increased search complexity, leading to only modest improvements.
In Figs.~\ref{fig:mmf_density}d, f, and g, the most significant throughput increases are observed when reducing the Eve density at the LEO layer, as LEO nodes enable source-to-user connections with fewer hops and higher spectral efficiency through inter-satellite routes.

\subsubsection{Throughput-secrecy gap}

Fig.~\ref{fig:mmsr_spsc} presents a comparative analysis between the max-min secrecy rate and the max-min throughput across SPSC probability thresholds. 
As the threshold $\tau$ increases, the two metrics exhibit a strong convergence, highlighting their equivalence under stringent security constraints. In contrast, in the regime of low SPSC thresholds, a noticeable performance gap occurs, primarily because the max-min throughput formulation does not intrinsically incorporate relay security considerations. 
However, maintaining a high SPSC probability (e.g., $>0.95$) is imperative to effectively mitigate the risk of eavesdropping by adversarial nodes such as Eve \cite{Yao18-TWC}.
Under the high-security regime, where the SPSC probability exceeds a predefined reliability threshold, the maximum throughput-secrecy gap is approximately 0.38\% at $1-\tau=10^{-3}$, which is marginal enough to justify the use of max-min throughput metric in highly secure scenarios.


\begin{figure*}[bht]
  \centering
  \includegraphics[width=\textwidth]{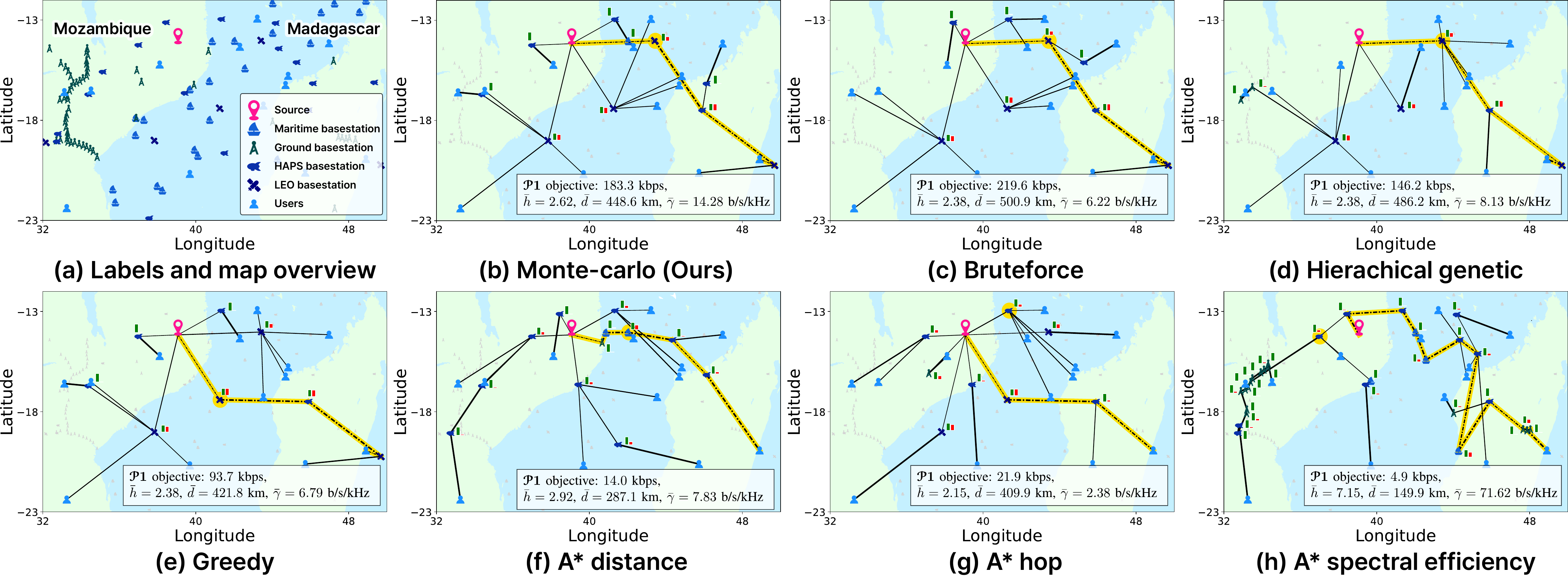}
  \vspace{-.5cm}
  \caption{Secure routing examples on the Mozambique-Madagascar testbed. 
  The highlighted dashed line represents the longest path. The highlighted node represents the min-throughput node.
  $\bar{h}$, $\bar{d}$, and $\bar{\gamma}$ represent average hop, link distance, and link spectral efficiency, respectively. The green and red bars adjacent to each node indicate the normalized transmission and jamming power allocated under the minimum transmission power ratio $P_i/P = 0.8$. The line width of each link represents the amount of allocated bandwidth.}
  \label{fig:real_demo}
\end{figure*}

\begin{figure*}[htb]
  \centering
  \includegraphics[width=\textwidth]{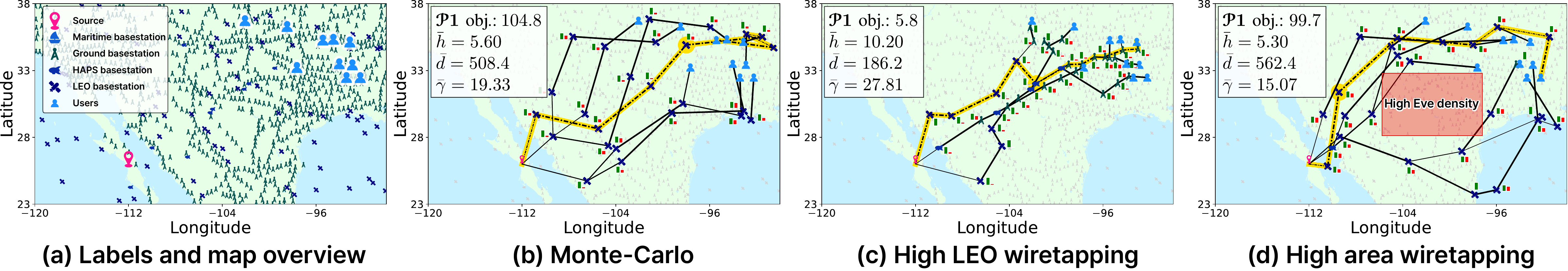}
  \vspace{-.5cm}
  \caption{
  Secure routing solutions on the Southern North America testbed. (a) Moderate uniform Eve density; (b) High Eve density in the LEO layer; (c) Localized high-density region (\textcolor[HTML]{EE2928}{\textbf{red area}}). Eve density in (b) and (c) is increased tenfold. Markers and annotations follow those in Fig.~\ref{fig:real_demo}.}
  \label{fig:real_demo_us}
\end{figure*}

\subsubsection{Impact of minimum transmit-power constraint}
Fig.~\ref{fig:mmf_power} examines how the max–min throughput varies as the minimum transmission power ratio $P^{\min}_i/P_i^{\max}$ changes.
For $P^{\min}_i/P_i^{\max}<0.65$, each node can flexibly allocate its transmission and jamming powers according to \eqref{eq:optimal_rho_sigma}, thereby maintaining a nearly constant achievable throughput.  

This finding suggests a practical design guideline as follows:
As $P_i^{\min}$ increases, the maximum available jamming budget gradually decreases, which shortens the admissible link distance under the SPSC constraint and excessively limits the feasible graph domain.  
Selecting $P^{\min}_i/P_i^{\max}$ within the range of $0.6$-$0.7$ ensures adequate jamming capability while eliminating non-promising links, thereby reducing the search space without compromising performance.

\subsection{Demonstrations on Real-World Data Testbeds}\label{subsec:demo}
We evaluate and visualize the framework in testbeds built from real-world base station data to understand scheme operations and framework behavior under secure environment changes.
We construct two testbeds in the Mozambique-Madagascar Channel and Southern North America.
This work is the first to establish \textbf{a SAGSIN testbed with a real-world dataset} integrating HAP base stations.
\hl{These testbeds reflect the large-scale HAP mobility patterns \cite{bellemare2020autonomous}}.\footnote{
Since no operational datasets for HAP base stations are currently available, the HAP nodes were synthesized with temporal adjustments.
}

\paragraph{Data preparation} The testbeds incorporate various network elements from multiple data sources, as detailed below:
\begin{itemize}
    \item \textbf{Ground.}
    Ground node positions are extracted from the open source project \href{https://www.opencellid.org}{\textsc{OpenCellID}} database, using a network snapshot at 2025-04-22 13:00\,UTC.
    \item \textbf{Maritime.}
    Maritime nodes for the Mozambique-Madagascar testbed are parsed from the \href{https://www.marinetraffic.com}{\textsc{MarineTraffic}} database at 2025-04-22 13:00\,UTC.
    \item \textbf{HAPs.}
    HAPs node positions are derived from the non-profit organization \href{https://stratocat.com.ar}{\textsc{Stratocat}} database based on data provided by Google's Project Loon \cite{bellemare2020autonomous}.
    Snapshots at 2020-09-28 10:00\,UTC for Mozambique-Madagascar; and at 2020-07-29 00:00\,UTC for Southern North America are used to represent HAPs deployments.
    \item \textbf{LEO satellites.} LEO node coordinates are extracted from the non-profit organization \href{https://celestrak.org}{\textsc{CelesTrak}} repository, parsing Starlink satellites at 2025-04-22 13:00\,UTC.
\end{itemize}
\hl{From the positional information obtained from the dataset, 
we generate channels as described in the system model using the simulation parameters in Table.~\ref{tab:simulation_parameters}.
}

\subsubsection{Mozambique-Madagascar Testbed}
Fig.~\ref{fig:real_demo} overlays the secure routes selected in the Mozambique–Madagascar testbed.  
In Fig.~\ref{fig:real_demo}b, the \textbf{Monte-Carlo} scheme consistently finds the “sweet‐spot” of two or three hops of a few hundred kilometers each, balancing the throughput penalty of extra relay hops with the signal loss incurred on longer links.  

The three A$^*$ metrics provide useful insights, but optimizing only one axis inevitably violates the inherent multi-dimensional trade-offs.
\textbf{A$^*$ hop} forces overly long links that suffer severe SNR penalties, while \textbf{A$^*$ spectral efficiency} fragments the path into many short hops and incur excessive scheduling overhead.  
These observations reveal that effective secure routing in SAGSIN must jointly account for hop count, link distance, and hop capacity.  

\subsubsection{Southern North America testbed}
Fig.~\ref{fig:real_demo_us} shows how the proposed MCRR adapts to varying Eve distributions in the Southern North America region.
By computing each link’s maximum allowable distance from Eve density and the SPSC threshold, the framework automatically adjusts routing when (i) density increases on a particular network layer (Fig.~\ref{fig:real_demo_us}c) and (ii) Eves concentrate in a specific geographic area (Fig.~\ref{fig:real_demo_us}d).
Thus, we can observe:
\begin{itemize}
  \item \textbf{Fig.~\ref{fig:real_demo_us}b:} When Eve density is moderate across all layers, the \textbf{Monte-Carlo} scheme adopts inter-satellite paths that maintain throughput in the hundreds of kbps.
  \item \textbf{Fig.~\ref{fig:real_demo_us}c:} If Eve density increases in a LEO network, the links in the network become infeasible. Then Alg.~\ref{alg:MCRR} detours through ground and HAPs nodes.
  \item \textbf{Fig.~\ref{fig:real_demo_us}d:} When Eves concentrate in the shaded area, links crossing that region are automatically excluded and replaced by detours around the area.
\end{itemize}
In summary, incorporating distance constraints computed from the Eve density and SPSC threshold into the optimization process enables routing decisions to promptly adapt to variations in node-specific and geographical secure threats.

\section{Discussion}
\paragraph{\hl{Conclusion}} This work demonstrates that secure and high-performance multi-hop communication can be achieved even in the presence of unknown eavesdroppers across large-scale SAGSINs.  
This achievement is made possible by deriving a closed-form expression for the SPSC probability and integrating it into a cross-layer optimization framework that jointly optimizes radio resources and relay routes.  
This framework features an $\mathcal{O}(1)$-complexity frequency allocation and power splitting strategy, along with a Monte-Carlo relay routing algorithm that ensures a minimum throughput for each user under security constraints.  
The proposed framework was validated through a real-world testbed incorporating ground, maritime, HAP, and LEO nodes, marking the first SAGSIN testbed that includes HAP base stations.  
This validation narrows the gap between theoretical models and practical deployment, offering promising insights into the realization of secure communication in future 6G integrated networks.

\paragraph{\hl{Limitations and Future Work}}
\hl{
While this paper proposes a novel approach to physical-layer secure routing in SAGSINs, several challenges remain open. 
First, the closed-form derivation of the SPSC probability in \eqref{eq:SPSC_final_approx_2} assumes Rayleigh fading. 
As links in SAGSINs can be LoS, particularly in the space and aerial layers, a derivation under Rician or Nakagami-$m$ fading models would be more appropriate, but remains an unsolved problem. 
Furthermore, challenges involving cross-layer or active attacks in SAGSINs, such as satellite jamming or ship-to-air interception, still need to be addressed. 
Although this work establishes a new paradigm for secure routing against unknown passive eavesdroppers, the system model remains limited to this threat type.
Thus, designing robust protocols against active and mobility-driven attacks in SAGSINs remains a key open problem.}

\section*{Acknowledgment}
We appreciate Prof. Lajos Hanzo for the valuable discussions on dataset construction and for kindly sharing the related research materials on SAGSINs.

\bibliographystyle{IEEEtran}
\bibliography{references.bib}

\newpage
\setcounter{page}{1}

\title{ {\sc \Large Supplementary Material of } \\ 
Secure Multi-Hop Relaying in Large-Scale Space-Air-Ground-Sea Integrated Networks}
\author{Hyeonsu Lyu,~\IEEEmembership{Student Member,~IEEE}, Hyeonho Noh,~\IEEEmembership{Member,~IEEE}, Hyun Jong Yang,~\IEEEmembership{Senior Member,~IEEE}, and Kaushik Chowdhury,~\IEEEmembership{Fellow,~IEEE}}
\maketitle

\begin{appendices}

{\center{ \sc \Large Contents }\\}

\hspace{12pt}

\noindent\fbox{\begin{minipage}{0.99\linewidth}
\textbf{Appendix~\ref{Supple:Notations and Variables}} Notations and Variables \dotfill \\
\textbf{Appendix~\ref{Supple:Approximation of Ergodic Spectral Efficiency}} Approximation of Ergodic Spectral Efficiency \dotfill \\
\textbf{Appendix~\ref{supple:SPSC Probability with Various Fading Effects}} SPSC Probability with Various Fading Effects \dotfill \\
\textbf{Appendix~\ref{Supple:Derivation of the SPSC Probability}} Derivation of SPSC Probability \dotfill \\
\textbf{Appendix~\ref{Supple:Derivation of the Bound Gap in the SPSC Approximation}} Approximation Gap in the SPSC Derivation \dotfill \\
\textbf{Appendix~\ref{Supple:Visualization of the MCRR Algorithm}} Visualization of the MCRR Algorithm \dotfill
\end{minipage}}

\newpage

\section{Notations and Variables}
\label{Supple:Notations and Variables}
All variables adopted in the paper can be summarized in the following table:

\begin{tablehl}[htbp]
    \centering
    \caption{Summary of Notations}
    \label{tab:Summary of Notations}
    \begin{tabularx}{\columnwidth}{>{$}l<{$} X}
        \toprule
        \textbf{Symbol} & \textbf{Description} \\
        \midrule
        \multicolumn{2}{l}{\textbf{Sets and Indices}} \\
        \mathcal{M}     & Set of indices for Eves. \\
        \mathcal{I}     & Set of indices for network nodes, $\{0, 1, ..., I\}$. \\
        \mathcal{U}     & Set of indices for users, $\{1, ..., U\}$. \\
        \mathcal{N}     & Set of nodes and users in the graph $\mathcal{G}$. \\
        \mathcal{E}     & Set of edges in the graph $\mathcal{G}$. \\
        \mathcal{E}_u   & Set of edges that constitute the relay path for user $u$. \\
        S               & A subset of nodes for the spanning tree constraint. \\
        \midrule
        \multicolumn{2}{l}{\textbf{Physical Layer Parameters}} \\
        \lambda_i       & Density of Eves in the network layer $i$. \\
        \bm{p}_i, \bm{p}_e        & Position vector of node $i$ and Eve $e$. \\
        d^{\mathsf{s}}_{(i,j)} & Distance between node $i$ and node $j$. \\
        d^{\mathsf{e}}_{i,e}   & Distance between node $i$ and Eve $e$. \\
        \alpha_i        & Path loss exponent for the link from node $i$. \\
        h^{\mathsf{s}}_{(i,j)}, h^{\mathsf{e}}_{i,e} & Small-scale Rayleigh fading channel coefficient for the legitimate link $(i,j)$ and the wiretap link $(i,e)$. \\
        \rho_i          & Transmit power density of node $i$. \\
        \sigma_i        & Jamming power density of node $i$. \\
        n_0             & Noise power spectral density. \\
        P_i^{\max},P_i^{\min}              & Max/min transmission power of a node $i$. \\
        \beta_{(i,j),u} & Bandwidth allocated to user $u$ on the link $(i,j)$. \\
        \midrule
        \multicolumn{2}{l}{\textbf{Performance Metrics}} \\
        \mathrm{SNR}^{\mathsf{s}}_{(i,j)} & SNR of the legitimate link $(i,j)$. \\
        \mathrm{SNR}^{\mathsf{e}}_{(i,e)}   & SNR of the wiretap link $(i,e)$. \\
        C_{(i,j)}       & Secrecy capacity of the link $(i,j)$. \\
        \mathbb{P}_{(i,j)} & SPSC Probability for the link $(i,j)$. \\
        \gamma_{(i,j)}  & Spectral efficiency of the link $(i,j)$. \\
        \eta_u          & Relay throughput for user $u$. \\
        h_u             & Number of hops in the relay path for user $u$. \\
        \midrule
        \multicolumn{2}{l}{\textbf{Graph and Routing Variables}} \\
        \mathcal{G}_{\mathrm{all}} & Directed network routing graph. \\
        x_{(i,j)}       & Binary variable indicating if edge $(i,j)$ is in the graph $\mathcal{G}$. \\
        x_{(i,j),u}     & Binary variable indicating if edge $(i,j)$ is in the relay path $\mathcal{E}_u$ for user $u$. \\
        \bottomrule
    \end{tabularx}
\end{tablehl}

\section{Approximation of Ergodic Spectral Efficiency}
\label{Supple:Approximation of Ergodic Spectral Efficiency}

\hl{
We approximate the ergodic spectral efficiency to capture long-term network throughput. 
By Jensen's inequality, this approximation provides an upper bound on the ergodic spectral efficiency. 
Figure~\ref{fig:ergodic_spectral_efficiency} validates the approximation by comparing the ergodic spectral efficiency with its no-fading counterpart. 
Over SNRs in the range [-40, 30]~dB, the approximation closely matches the true curve in the low-SNR regime, and the gap widens as SNR increases. 
This gap can be reduced by introducing a scaling factor $\alpha$ that multiplies~\eqref{eq:spectral_efficiency}. A least-squares fit yields $\alpha = 0.8908$, achieving a mean squared error (MSE) of $0.008$ over the specified SNR range. Hence, the approximation in~\eqref{eq:spectral_efficiency} does not materially affect the decision variables, and any residual bias can be compensated by this simple linear rescaling.}

\begin{figurehl}[htbp]
    \centering
    \includegraphics[width=\linewidth]{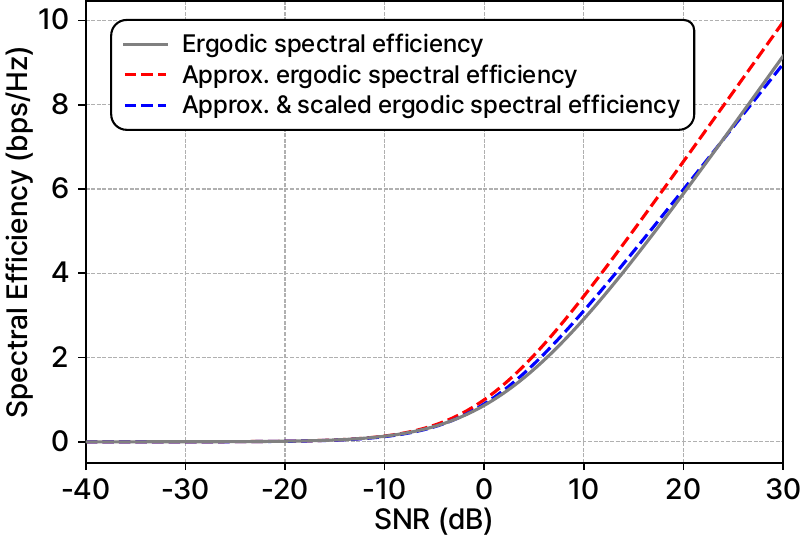}
    \caption{Comparision between ergodic spectral effieciency and its approximation}
    \label{fig:ergodic_spectral_efficiency}
\end{figurehl}

\section{SPSC Probability with Various Fading Effects}
\label{supple:SPSC Probability with Various Fading Effects}
\hlframe{
Links in SAGSINs can exhibit a strong LoS or specular component. 
Then, assuming pure Rayleigh can therefore be overly pessimistic for such links. 
) components and their slow variations are modeled using Rician and Shadowed-Rician fading, which are canonical choices in SAGSIN studies.
The channel power are normalized to ensure consistent comparisons across other fading models.
Parameters $K_{\mathrm{dB}}=8$ for a Rician fading and $m=3$ for a shadowed-Rician fading are adopted as in the 3GPP standard [1].

Figure~\ref{fig:spsc_density_rician_fading} provides the SPSC probability obtained from Monte-Carlo simulations and \eqref{eq:SPSC_final_approx_2}.
Simulation parameters are identically configured as in Sec.~\ref{subsec:SPSC Probability Analysis}.
This indicates that the Rayleigh fading provides worst-case SPSC probability since introducing LoS component helps the legitimate link more than the eavesdropper.
(Shadowed) Rician fading disproportionately benefits the short legitimate link, yielding a larger performance gain than for typically distant eavesdroppers.

Moreover, we observe that the change of fading primarily rescales the Eve density and calibration parameters in \eqref{eq:calibarated_SPSC_approximation}, but preserves the curve’s shape.
Equivalently, transitioning from Rayleigh to (shadowed) Rician fading shifts the SPSC curve rightward, effectively resembling a reduction in eavesdropper density.
Thus, Rayleigh offers a worst-case guarantee, while (shadowed) Rician represent realistic regimes. 
}

\begin{figurehl}[htb]
    \centering
    \includegraphics[width=1\linewidth]{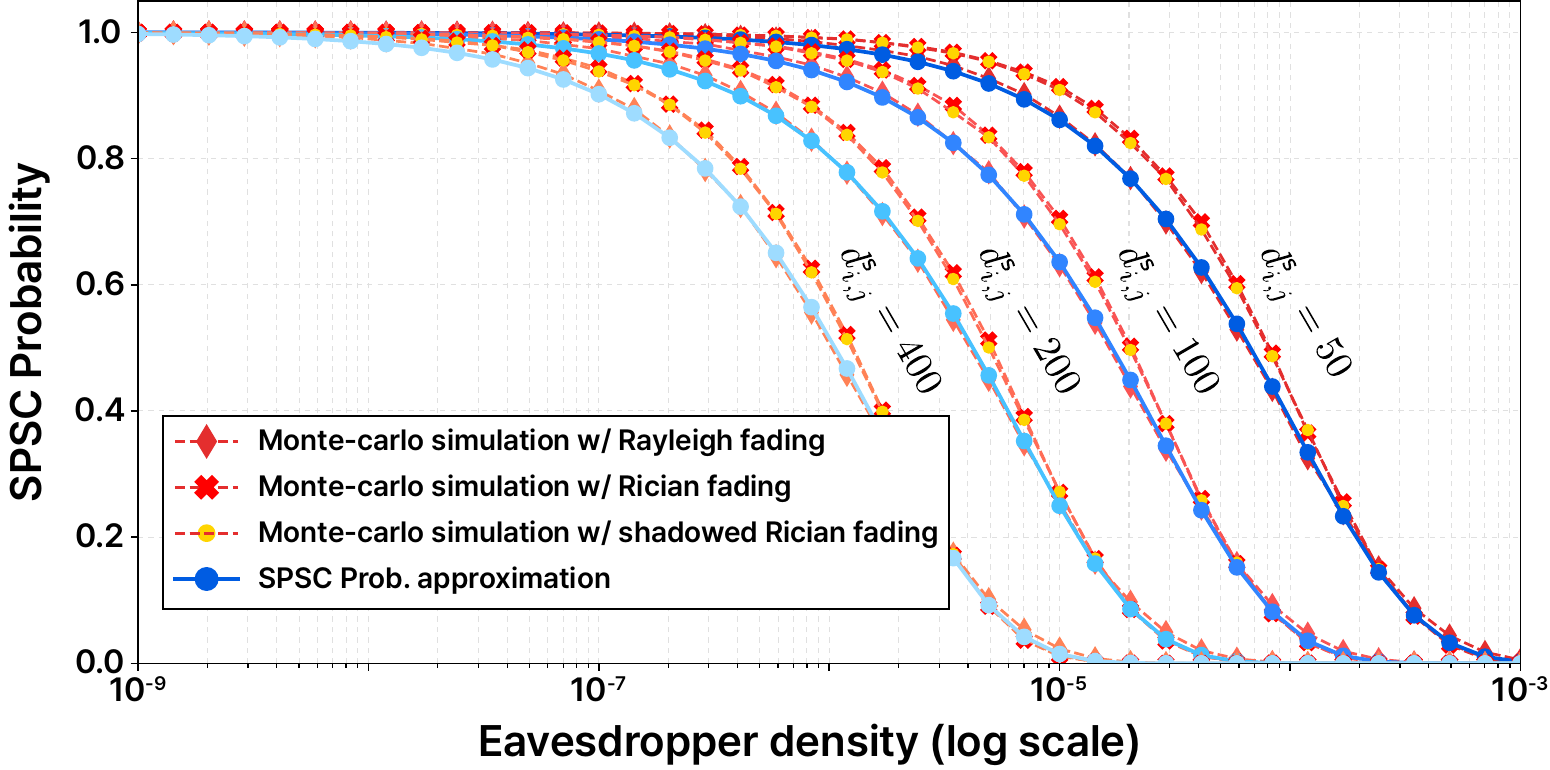}
    \caption{The SPSC probability versus Eve density for various link distances and fading effects}
    \label{fig:spsc_density_rician_fading}
\end{figurehl}

\onecolumn
\hlframe{
\section{Derivation of the SPSC Probability}
\label{Supple:Derivation of the SPSC Probability}

Using $|h^\mathsf{e}_{(i,e)}|^2 \sim \mathsf{Exp}(1)$, \eqref{eq:E1_derivation_2} can be simplified as follows:
\begin{align}
    &\!\!\mathbb{P}_{(i,j)} \hspace{-0.05cm}= \hspace{-0.05cm}\mathbb{E}_{|h|^2}\!
    \bigg[\!
        \exp\!\Big[\!
            -\!\lambda_i\hspace{-0.1cm}\int_{\mathbb{R}^2}\hspace{-0.15cm}\mathbb{P}\big(\SNR{s}_{(i,j)}<\SNR{e}_{(i,e)}\big) ~d\bm{p}_e
        \!\Big]
    \!\bigg] 
    \hspace{-.25cm}
    \\[-.1cm]
    &\!\!\! = \hspace{-0.05cm}\mathbb{E}_{|h|^2}\!
    \bigg[\!
        \exp\!\Big[\! 
    -\!\lambda_i\hspace{-0.1cm}\int_{\mathbb{R}^2}\hspace{-0.15cm}\mathbb{P}
                \bigg(
                    \frac{
                        |h^\mathsf{s}_{(i,j)}|^2 (d^\mathsf{e}_{(i,e)})^{\alpha_i} n_0 
                    }{
                        (d^\mathsf{s}_{(i,j)})^{\alpha_i}n_0  - \sigma_i G_{(i,j)}  |h^\mathsf{s}_{(i,j)}|^2
                    }
                    <  |h^\mathsf{e}_{(i,e)}|^2
                \bigg) d\bm{p}_e
        \!\Big]
    \!\bigg] 
    \\     
    \label{eq:derivation_1}
    &\!\!\! = \hspace{-0.05cm}\mathbb{E}_{|h|^2}\!
    \bigg[\!
        \exp\!\Big[\!
    -\!\lambda_i\hspace{-0.1cm}
        \underbrace{ 
        \int_{\mathbb{R}^2}\hspace{-0.15cm}
            \exp
                \bigg(
                    -\frac{
                        |h^\mathsf{s}_{(i,j)}|^2 (d^\mathsf{e}_{(i,e)} n_0)^{\alpha_i}
                    }{
                        (d^\mathsf{s}_{(i,j)})^{\alpha_i} n_0 - \sigma_i G_{(i,j)}  |h^\mathsf{s}_{(i,j)}|^2
                    }
                \bigg) d\bm{p}_e
        }_{\displaystyle \textcolor{RoyalBlue}{\textbf{(b)}}}
        \!\Big]
    \!\bigg]. 
\end{align}
Then, we can apply a polar coordinate transform to (a), which give a direct integral as
\begin{align}
    \textcolor{RoyalBlue}{\textbf{(b)}} &= \int_{\mathbb{R}^2}\hspace{-0.15cm}
            \exp
                \bigg(
                    -\frac{
                        |h^\mathsf{s}_{(i,j)}|^2 (d^\mathsf{e}_{(i,e)} n_0)^{\alpha_i}
                    }{
                        (d^\mathsf{s}_{(i,j)})^{\alpha_i} n_0 - \sigma_i G_{(i,j)}  |h^\mathsf{s}_{(i,j)}|^2
                    }
                \bigg) d\bm{p}_e \\
        &= \int_0^\infty \hspace{-0.15cm}
            r \exp
                \bigg(
                    -\frac{
                        |h^\mathsf{s}_{(i,j)}|^2 r^{\alpha_i} n_0
                    }{
                        (d^\mathsf{s}_{(i,j)})^{\alpha_i} n_0 - \sigma_i G_{(i,j)}  |h^\mathsf{s}_{(i,j)}|^2
                    }
                \bigg) dr \\
        &= \frac{2\pi}{\alpha_i}\Gamma(\frac{2}{\alpha_i})
        \bigg(
            -\frac{
                |h^\mathsf{s}_{(i,j)}|^2 n_0
            }{
                (d^\mathsf{s}_{(i,j)})^{\alpha_i} n_0 - \sigma_i G_{(i,j)}  |h^\mathsf{s}_{(i,j)}|^2
            }
        \bigg)^{-\frac{2}{\alpha_i}} \\
        &= \frac{2\pi}{\alpha_i}\Gamma(\frac{2}{\alpha_i})(d^\mathsf{s}_{(i,j)})^{2}
        \bigg(
            \frac{
                |h^\mathsf{s}_{(i,j)}|^2
            }{
                1  - \frac{\sigma_i G_{(i,j)}  }{(d^\mathsf{s}_{(i,j)})^{\alpha_i} n_0}
                |h^\mathsf{s}_{(i,j)}|^2
            }
        \bigg)^{-\frac{2}{\alpha_i}}.
        \label{eq:exact_phi_h}
\end{align}
Plugging in (a) into \eqref{eq:derivation_1} and applying Jensen's inequality gives
\begin{align}
    \mathbb{P}_{(i,j)} &\approx 
        \exp\big[
            \mathbb{E}_{|h|^2}[(a)
        ]\big] \\
    &= \exp\bigg[
            -\lambda_i \frac{2\pi}{\alpha_i} 
            \Gamma(\frac{2}{\alpha_i}) (d^\mathsf{s}_{(i,j)})^2
        \underbrace{
            \mathbb{E}_{|h|^2}
                \bigg[
                \bigg(
                    \frac{
                        |h^\mathsf{s}_{(i,j)}|^2
                    }{
                        1  - 
                        \frac{\sigma_i G_{(i,j)}}{(d^\mathsf{s}_{(i,j)})^{\alpha_i} n_0}
                        |h^\mathsf{s}_{(i,j)}|^2
                    }
                \bigg)^{-\frac{2}{\alpha_i}}
            }_{\displaystyle \textcolor{RoyalBlue}{\textbf{(c)}}}
        \bigg]
\end{align}

Finally, applying the binomial approximation on the denominator of \textbf{(b)} results in
\begin{align}
    \textcolor{RoyalBlue}{\textbf{(c)}} &= \mathbb{E}_{|h|^2}
            \bigg[
                 |h^\mathsf{s}_{(i,j)}|^{-\frac{4}{\alpha_i}}
                 \Big(
                    1  - \frac{\sigma_i G_{(i,j)}}{(d^\mathsf{s}_{(i,j)})^{\alpha_i} n_0}
                    |h^\mathsf{s}_{(i,j)}|^2
                 \Big)^{\frac{2}{\alpha_i}}
            \bigg]
    \nonumber\\
    &\approx  \mathbb{E}_{|h|^2}
            \bigg[
                |h^\mathsf{s}_{(i,j)}|^{-\frac{4}{\alpha_i}}
                \Big(
                    1  - \frac{2 \sigma_i G_{(i,j)}}{\alpha_i (d^\mathsf{s}_{(i,j)})^{\alpha_i} n_0}
                    |h^\mathsf{s}_{(i,j)}|^2
                \Big)
            \bigg]
    \nonumber\\
    &= \mathbb{E}_{|h|^2}
        \big[
            |h^\mathsf{s}_{(i,j)}|^{-\frac{4}{\alpha_i}}
        \big]
        -
        \frac{2 \sigma_i G_{(i,j)}}
             {\alpha_i (d^\mathsf{s}_{(i,j)})^{\alpha_i} n_0}
        \mathbb{E}_{|h|^2}
        \Big[
            |h^\mathsf{s}_{(i,j)}|^{2-\frac{4}{\alpha_i}}
        \Big]
    \nonumber\\
    &= \Gamma \big(1-\frac{2}{\alpha_i} \big) 
    -
    \frac{2 \sigma_i G_{(i,j)}}
             {\alpha_i (d^\mathsf{s}_{(i,j)})^{\alpha_i} n_0}
    \Gamma \big(2-\frac{2}{\alpha_i} \big)
    \label{eq:approximiation_b}
\end{align}
Then, the SPSC probabilty can be approximately derived as
\begin{align}
    \mathbb{P}_{(i,j)}=
    \exp
    \bigg[
        -\kappa_i
        \Big[
        \Gamma \big(1-\frac{2}{\alpha_i} \big) 
        -
        \frac{2 \sigma_i G_{(i,j)}}
                 {\alpha_i (d^\mathsf{s}_{(i,j)})^{\alpha_i} n_0}
        \Gamma \big(2-\frac{2}{\alpha_i} \big)
        \Big]
        (d^\mathsf{s}_{(i,j)})^2
    \bigg]
\end{align}
for $\kappa_i=\lambda_i \frac{2\pi}{\alpha_i}\Gamma(\frac{2}{\alpha_i})$.
When $\sigma_i=0$, this approximation reduces to the previously reported result in \cite{Yao16-TCOM}, demonstrating both the validity and extensibility of the derivation.

\section{Derivation of the Bound Gap in the SPSC Approximation}
\label{Supple:Derivation of the Bound Gap in the SPSC Approximation}

\subsection{High Pathloss Exponent Case}
For brevity of notations, we define 
\begin{align}
\Phi(|h|^2) 
&:= 
\frac{2\pi}{\alpha_i}\Gamma(\frac{2}{\alpha_i})(d^\mathsf{s}_{(i,j)})^{2}
        \bigg(
            \frac{
                |h|^2
            }{
                1  - \frac{\sigma_i G_{(i,j)}} {(d^\mathsf{s}_{(i,j)})^{\alpha_i} n_0}
                |h|^2
            }
        \bigg)^{-\frac{2}{\alpha_i}}
\end{align}
which is a non-negative random variable depending on the fading coefficient $|h|^2 \sim \mathsf{exp}(1)$.
With this notation, the original success probability is written as
\begin{align}
\mathbb{P}_{(i,j)}
&= \mathbb{E}_{|h|^2}\left[\exp(-\lambda_i \Phi(h))\right]. \label{eq:main_def}
\end{align}
Since the exponential function is convex, Jensen's inequality yields
\begin{align}
\mathbb{E}_{|h|^2}\big[\exp(-\lambda_i \Phi(|h|^2))\big]
\geq \tilde{\mathbb{P}}_{(i,j)} := \exp\left(-\lambda_i \mu \right),
\label{eq:LB}
\end{align}
for $\mu := \mathbb{E}_{|h|^2}[\Phi(|h|^2)]$.

Define the Jensen gap as
\begin{align}
\Delta := \mathbb{P}_{(i,j)} - \tilde{\mathbb{P}}_{(i,j)} \geq 0.
\end{align}
Let $\varphi(x)=e^{-\lambda_i x}$, which is convex on $[0,\infty)$ with
$\varphi''(x)=\lambda_i^2 e^{-\lambda_i x}\le \lambda_i^2$.
A standard bound on Jensen gaps for convex functions yields
\begin{align}
0 &\le \Delta
= \mathbb{E}_h[\varphi(\Phi)]-\varphi(\mu)
\\
& \le \frac{1}{2}\sup_{x\ge 0}\varphi''(x)\mathrm{Var}(\Phi)
\le \frac{\lambda_i^2}{2}\mathrm{Var}(\Phi).
\label{eq:delta_bound}
\end{align}
The bound in \eqref{eq:delta_bound} holds whenever $\mathrm{Var}(\Phi)<\infty$.

For any $m$ with $\frac{2m}{\alpha_i}<1$, the $m$-th moment exists and equals
\begin{align}
\mathbb{E}\big[\Phi^m\big]
    &= 
    \bigg(
    \frac{2\pi}{\alpha_i}\Gamma(\frac{2}{\alpha_i})(d^\mathsf{s}_{(i,j)})^{2}
    \bigg)^m
        \bigg(
            \frac{
                |h|^2
            }{
                1  - \frac{\sigma_i G_{(i,j)}} {(d^\mathsf{s}_{(i,j)})^{\alpha_i} n_0}
                |h|^2
            }
        \bigg)^{-\frac{2m}{\alpha_i}} \\
    &= \bigg(
    \frac{2\pi}{\alpha_i}\Gamma(\frac{2}{\alpha_i})(d^\mathsf{s}_{(i,j)})^{2}
    \bigg)^m
    \bigg[
    \Gamma \big(1-\frac{2m}{\alpha_i} \big) 
    -
    \frac{2m \sigma_i G_{(i,j)}}
             {\alpha_i (d^\mathsf{s}_{(i,j)})^{\alpha_i} n_0}
    \Gamma \big(2-\frac{2m}{\alpha_i} \big)
    \bigg]
\end{align}
Then, we have
\begin{align}
&\mu=
    \frac{2\pi}{\alpha_i}\Gamma(\frac{2}{\alpha_i})(d^\mathsf{s}_{(i,j)})^{2}
    \bigg[
    \Gamma \big(1-\frac{2}{\alpha_i} \big) 
    -
    \frac{2 \sigma_i G_{(i,j)}}
             {\alpha_i (d^\mathsf{s}_{(i,j)})^{\alpha_i} n_0}
    \Gamma \big(2-\frac{2}{\alpha_i} \big)
    \bigg]
\\
&\mathrm{Var}(\Phi)
= \bigg(
    \frac{2\pi}{\alpha_i}\Gamma(\frac{2}{\alpha_i})(d^\mathsf{s}_{(i,j)})^{2}
    \bigg)^2
    \bigg[
    \Gamma \big(1-\frac{4}{\alpha_i} \big) 
    -
    \frac{4 \sigma_i G_{(i,j)}}
             {\alpha_i (d^\mathsf{s}_{(i,j)})^{\alpha_i} n_0}
    \Gamma \big(2-\frac{4}{\alpha_i} \big)
    \bigg].
\label{eq:Phi_mean_variance}
\end{align}
Plugging \eqref{eq:Phi_mean_variance} in \eqref{eq:delta_bound} provides a tight Jensen bound when $\alpha_i > 4$.

\subsection{Low Pathloss Exponent Case}
When $\alpha_i < 4$, the variance of $\Phi$ diverges since the integral \eqref{eq:exact_phi_h} assumes that Eves are distributed over an infinitely large space.
Instead, we restrict the integration domain to $(0,R)$ to compute a more accurate gap.
As $\Delta$ increases with $\alpha_i$, \eqref{eq:delta_bound} is bounded by $\mathrm{Var}(\Phi)$ at $\alpha_i=4$.

Then, for $C=\Big[(d^\mathsf{s}_{(i,j)})^{\alpha_i} n_0 - \sigma_i G_{(i,j)}  |h^\mathsf{s}_{(i,j)}|^2 \Big]^{-1}$,
we have
\begin{align}
\mathbb{E}_{|h|^2}[\Phi]
&= 2\pi \int_0^R r \mathbb{E}_{|h|^2}\big[\exp\left(
-C|h|^2 r^4
\right)\big]dr \\
&= \pi \int_0^R \frac{r}{1+Cr^4} dr\\
&= \frac{\pi}{\sqrt{C}}\arctan(R^2\sqrt{C}) \approx \frac{\pi^2}{2\sqrt{C}}
\end{align}
for sufficiently large $R$.

Also, for variance, we have
\begin{align}
    \mathbb{E}_{|h|^2}[\Phi^2]
    &=\mathbb{E}_{|h|^2}\left[
    \Big(
     2\pi \int_0^R r \exp\left(
-C|h|^2 r^4
\right)dr
    \Big)^2
    \right] 
    \\
    &=\pi^2
    \int_0^{R^2}\!\!\!\int_0^{R^2}
    \mathbb{E}_{|h|^2}\Big[
    \exp\left(
-C|h|^2 (u^2+s^2)
\right)\Big]duds \\
    &=\pi^2
    \int_0^{R^2}\!\!\!\int_0^{R^2}
    \frac{duds}
    {1+C(r^4_1+r^4_2)} \\
    &\leq
    \pi^2\int_0^{2\pi}\!\!\!\int_0^{\sqrt{2}R^2}
    \frac{r}
    {1+Cr^2}drd\theta \\
    &=\frac{\pi^3}{4C}\ln(1+2CR^4)
    \approx \frac{\pi^2}{C}\ln{R}.
\end{align}
Thus, we have 
\begin{align}
    \mathrm{Var}(\Phi) = 
    \frac{\pi^3}{C}\ln{R} - \frac{\pi^4}{4C}.
    \label{eq:Phi_approx_variance}
\end{align}
Then, plugging \eqref{eq:Phi_approx_variance} in \eqref{eq:delta_bound} provides an asymptotic bound for $\alpha_i < 4$.
}

\clearpage
\section{Visualization of the MCRR Algorithm}
\label{Supple:Visualization of the MCRR Algorithm}
\begin{figurehl}[htb]
    \centering
    \includegraphics[width=0.96\linewidth]{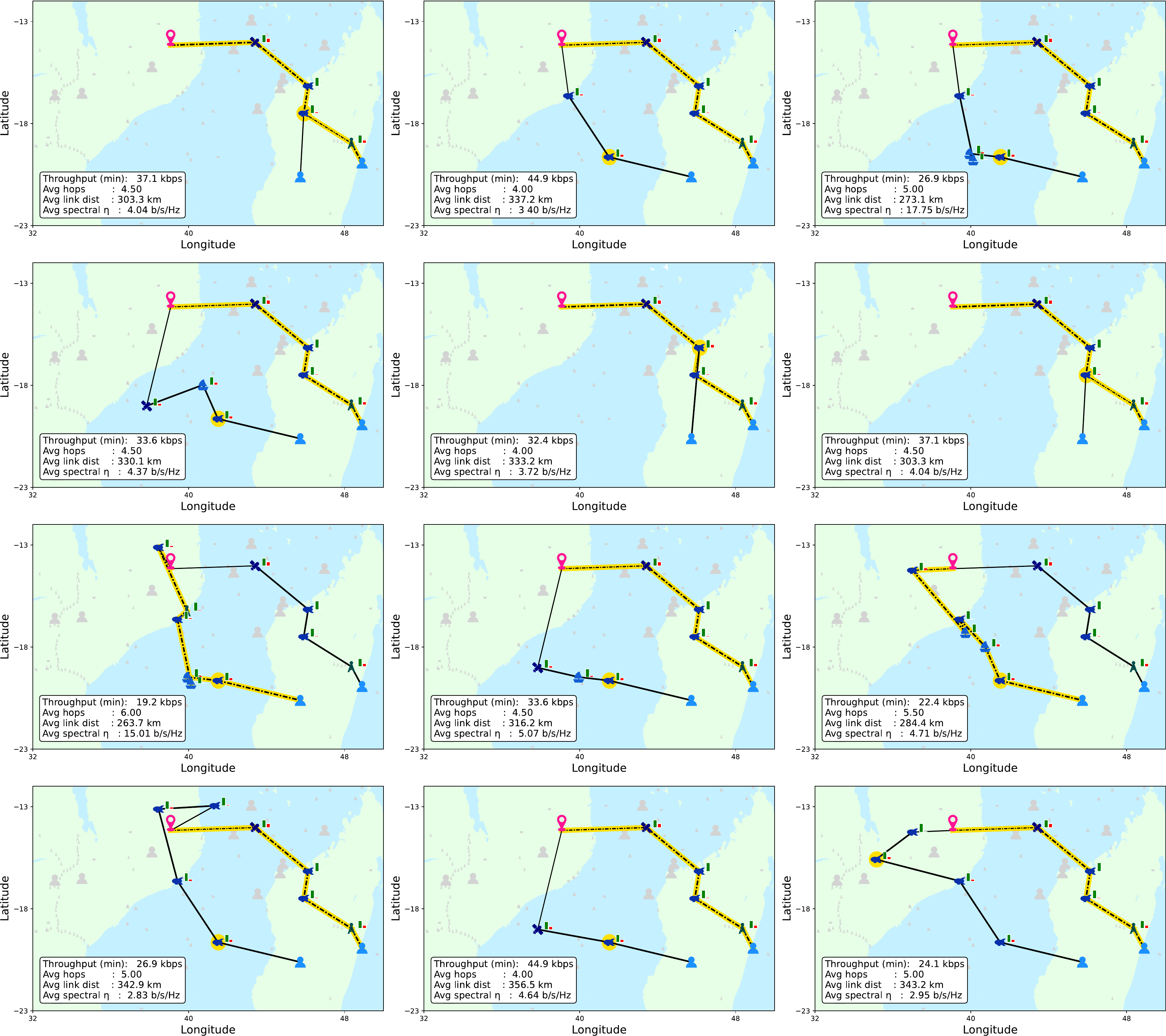}
    \caption{Visualization of twelve biased random walk samples of MCRR algorithm implemented in Mozambique-Madagascar testbed.
    Among them, a path with the largest min throughput is selected for the user.
    The highlighted dashed line represents the longest path. The highlighted node represents the min-throughput node.}
    \label{fig:mcrr_candidate_path_visualization}
\end{figurehl}

\begin{figurehl}[htb]
    \centering
    \includegraphics[width=0.96\linewidth]{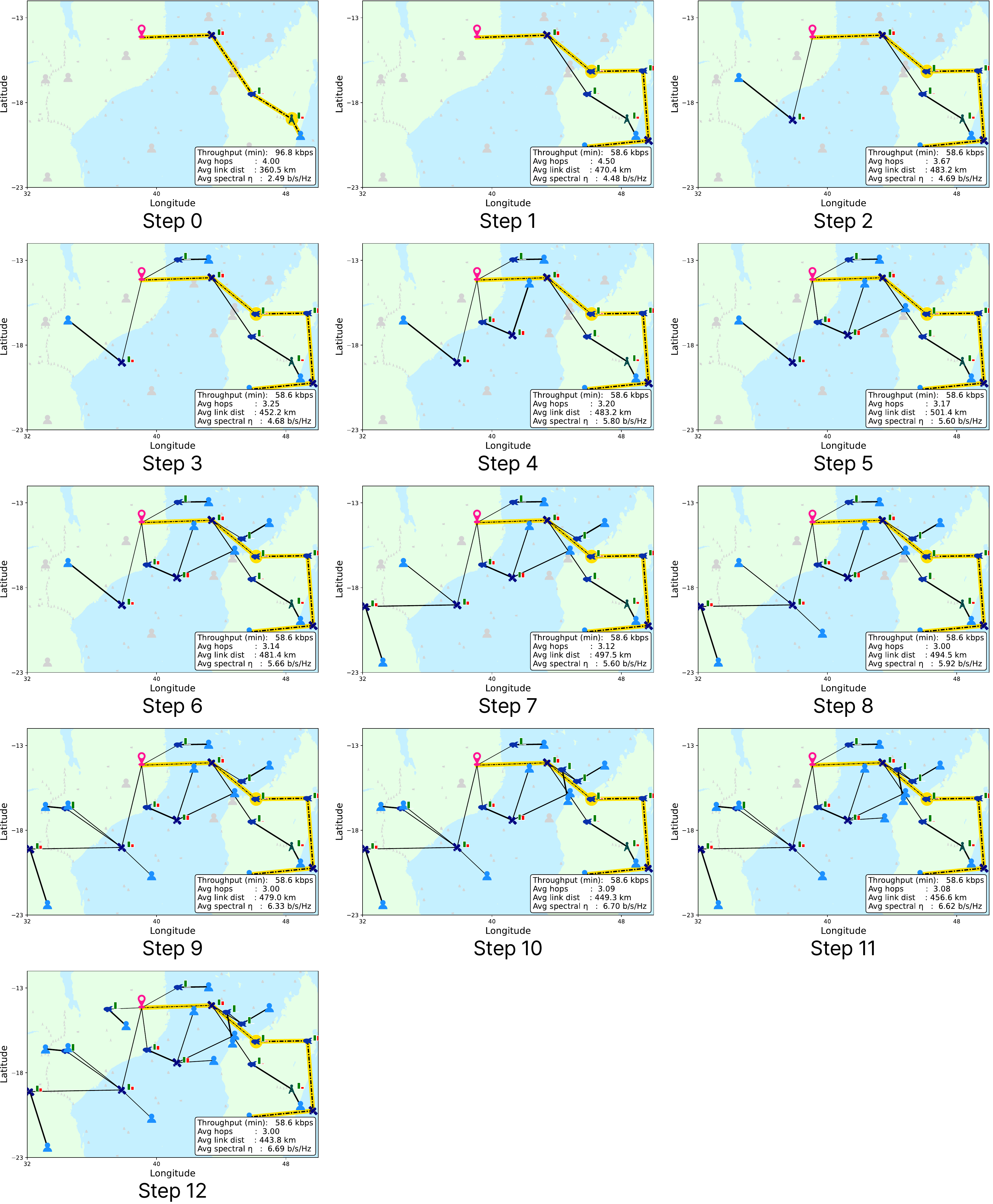}
    \caption{Step-by-step visualization of the MCRR algorithm implemented in Mozambique-Madagascar testbed. The highlighted dashed line represents the longest path. The highlighted node represents the min-throughput node.}
    \label{fig:mcrr_visualization}
\end{figurehl}

\end{appendices}
\end{document}